\documentclass[11pt,twocolumn]{article}
\usepackage[margin=0.75in,columnsep=0.25in]{geometry}
\usepackage{booktabs}
\usepackage{amsmath,amsthm,amssymb}
\usepackage{mathtools}
\usepackage{natbib}
\usepackage[hidelinks]{hyperref}
\usepackage{graphicx}
\usepackage{placeins}

\newtheorem{assumption}{Assumption}
\newtheorem{theorem}{Theorem}
\newtheorem{proposition}[theorem]{Proposition}
\newtheorem{corollary}[theorem]{Corollary}

\theoremstyle{remark}

\newcommand{\tS}{\tau_{S}}
\newcommand{\tY}{\tau_{Y}}
\newcommand{\tSi}{\tau_{S,i}}
\newcommand{\tYi}{\tau_{Y,i}}
\newcommand{\Om}{\Omega}
\newcommand{\Oss}{\Omega_{SS}}
\newcommand{\Osy}{\Omega_{SY}}
\newcommand{\Oyy}{\Omega_{YY}}
\newcommand{\RC}{R^2_{\mathcal{C}}}
\newcommand{\E}{\mathbb{E}}
\newcommand{\Prob}{\mathbb{P}}
\newcommand{\Var}{\mathrm{Var}}
\newcommand{\Cov}{\mathrm{Cov}}
\newcommand{\Corr}{\mathrm{Corr}}

\setcitestyle{numbers,square,comma}

\begin{document}

\title{When Can You Ship on Evals Alone? Trial-Level Surrogacy for Offline
Evaluation of LLM Systems}
\author{M\r{a}rten Schultzberg \\ Spotify \\ Stockholm, Sweden \\ \texttt{mschultzberg@spotify.com}}
\date{}

\makeatletter
\twocolumn[{%
\begin{@twocolumnfalse}
\maketitle
\begin{abstract}
\noindent
Offline evals increasingly guide changes to LLM-powered products, but showing that an eval agrees with human judgment does not tell a team which changes it can ship without an A/B test. That decision depends on whether eval treatment effects predict online treatment effects across a class of interventions. We call this requirement trial-level surrogacy and show that it is logically independent of unit-level (Prentice) surrogacy. In historical launch logs, sampling noise in both effect estimates attenuates the observed relationship. Under a bivariate normal working model, subtracting the reported within-intervention variances recovers the between-intervention covariance, and the corrected moments give a closed-form ship/test/kill rule that bounds the probability of shipping a change with a non-positive online effect. LLM judges add a separate distortion. When the judge is equally accurate on both arms, eval effects are only rescaled, but arm-dependent accuracy can reverse their sign. Once deployed, the eval-to-outcome relationship can no longer be observed in the ship and kill regions, so drift there becomes undetectable without randomized audit experiments. On 183 interventions from the Upworthy Research Archive, the pooled eval-to-outcome relationship is weak despite precise measurement on both sides, and the illustrative gate sends every intervention to an A/B test.
\end{abstract}
\vspace{1.5em}
\end{@twocolumnfalse}
}]
\makeatother

\section{Introduction}
\label{sec:intro}

Offline evals are increasingly used to iterate on LLM-powered products, but a good eval is not automatically a sufficient basis for shipping. Standard validation asks whether judge outputs agree with human labels or whether benchmark scores agree with external assessments across systems \cite{shankar2024,husainshankar2026,anthropic2026,zheng2023,jimenez2024,yao2024,arora2025}. A shipping decision asks a different question. After a prompt revision, retrieval change, ranking adjustment, or model swap, does the \emph{change} in the offline eval predict the \emph{change} in the online outcome? Recent production work has begun to compare LLM-judge eval deltas with online business outcomes across deployed variants \cite{gupta2026}. We study how to construct a gate that decides, based on offline evals, for each change, whether to ship it, kill it, or send it to an A/B test. The statistical object such a gate needs is the relationship between treatment effects across interventions, rather than agreement at the output or system level.

We therefore take the intervention as the unit of analysis. Intervention $i$ has a true eval effect $\tSi$ and a true online effect $\tYi$. Within a prespecified intervention class $\mathcal{C}$, the eval is useful for shipping only insofar as its effects predict online effects across interventions. This is trial-level surrogacy in the sense of the clinical-trials literature \cite{danielshughes1997,buyse2000}, with each intervention in the role of a trial. Say a product team wants to limit the risk of harm to its online metrics when it makes decisions based on offline eval deltas. For a new intervention it observes a noisy eval delta $X$ but not $\tY$, so it needs the probability that the online effect has the wrong sign given that delta, $\Prob(\tY\le0\mid X=x)$ for a ship and $\Prob(\tY>0\mid X=x)$ for a kill. These risks depend on the intervention class, the eval, and the online outcome. They are not an intrinsic property of the eval, and learning them requires that past interventions and the new one come from the same class. Proposition~\ref{prop:levels} shows why conventional notions of eval validity do not answer this question. Unit-level Prentice surrogacy can hold while the eval effect says nothing about the sign of the online effect, and the reverse can also occur.

\paragraph{Contributions.}
The paper makes three connected contributions. First, it states what any rule that ships or kills on eval evidence alone must learn from past launches. The rule needs the relationship between eval treatment effects and online treatment effects across a class of comparable interventions, and both measurement noise and LLM-judge error interfere with learning that relationship. Second, it shows that once a gate starts shipping changes on eval evidence alone, the eval-to-outcome relationship can no longer be observed in the ship and kill regions, and randomized audit experiments restore observability. Third, under simplifying assumptions, it works through the estimation of the gate end to end, to build intuition for what each assumption buys and where bias enters, and derives a closed-form ship/test/kill rule.

\paragraph{Relation to prior work.}
The distinction between individual-level and trial-level surrogate validity is established in the clinical-trials literature. Biomarkers that correlate with survival across patients need not predict treatment effects on survival \cite{prentice1989,buyse1998,buyse2000}, and intuitive surrogate criteria can admit the surrogate paradox \cite{chengengjia2007,vanderweele2013}. \citet{atheychetty2019} combine short-term proxies into a surrogate index at the unit level, while \citet{persson2026} study LLMs as synthetic experimental participants. Our target is instead whether eval treatment effects predict online treatment effects across a corpus of product interventions.

To estimate the eval-to-outcome relationship from noisy launch logs, we subtract the reported within-intervention measurement variances from the sample covariance. This correction follows the random-effects and errors-in-variables tradition. Subtracting within-study variances goes back to \citet{dersimonian1986}, with bivariate and multivariate extensions in \citet{danielshughes1997,riley2008,jackson2011}. Related algebra appears in attenuation correction \cite{spearman1904}, misclassification correction \cite{rogangladen1978}, and errors-in-variables regression \cite{fuller1987,carroll2006}. \citet{tripuraneni2024} study noise-aware proxy metric selection from A/B-test corpora \cite{deng2016,kohavi2020}. Supplementary analyses connect the selected-corpus problem to Pearson--Lawley selection formulas \cite{lawley1943,muthen1989,muthen2015}, the closed-loop problem to selective labels \cite{lakkaraju2017}, and randomized audits to the logic of long-term holdouts and exploit--explore designs \cite{hohnhold2015,sutton2018}.

\section{Trial-Level Surrogacy and Eval-Only Decisions}
\label{sec:setup}

\subsection{What the shipping decision needs}
The shipping decision concerns two intervention-level quantities, the effect on the eval and the effect on the online outcome. If the eval carries no information about the online effect, it cannot help with the decision. If it does, the relevant object is the conditional distribution of the unobserved online effect given the eval delta we actually observe.

Throughout, the unit of analysis is an intervention to an AI system. Intervention $i$ has two true effects. The \emph{online effect} $\tYi$ is the average treatment effect on the primary user outcome (e.g.\ engagement, revenue, resolution rate) that a sufficiently large A/B test would measure. The \emph{eval effect} $\tSi$ is defined through the eval suite. Tasks $t$ are drawn from a task distribution $\mathcal{Q}$, the system's output on task $t$ under arm $a\in\{0,1\}$ has latent binary quality $Z_a(t)\in\{0,1\}$, and
$\tSi=\E_{\mathcal{Q}}[Z_1(t)-Z_0(t)]$. The framework generalizes to a vector of $K$ graders with finite-variance scores using a surrogate index \cite{atheychetty2019}. As a running example, consider an LLM support agent. The eval judges whether each sampled reply resolves the ticket correctly under policy, while the online outcome is true resolution without hand-off to a human or re-contact within seven days.

For a new intervention from class $\mathcal{C}$, the team runs the eval at a given precision and observes the delta $X=\tS+\epsilon$, where $\epsilon$ is finite-sample error with variance $\sigma_S^2$, and does not observe $\tY$. Shipping is a mistake if $\tY\le0$, and killing is a mistake if $\tY>0$. Given what the team sees, the probabilities of these mistakes are
\begin{equation}
\Prob(\tY\le0\mid X=x,\mathcal{C})
\quad\text{and}\quad
\Prob(\tY>0\mid X=x,\mathcal{C}).
\label{eq:target}
\end{equation}
These are the risks a gate, a rule that routes each change to ship, kill, or A/B test, has to control. We ship only where the first probability is at most $\alpha$, kill only where the second is at most $\alpha$, and test otherwise. The safe regions need not be intervals in general. The threshold form used later comes from the monotone working model in Section~\ref{sec:model}. Neither probability is observed for a new intervention, so both have to be learned from earlier interventions in the same class. We suppress $\mathcal{C}$ below when the class is fixed.

Standard eval validation produces evidence at other levels. \emph{Output-level} validity asks whether judge labels agree with human labels on individual model outputs. \emph{Cross-system} validity asks whether the eval ranks a set of heterogeneous models the same way human judgments do. Neither links intervention-level changes in the eval to changes in the online outcome. A third concept, \emph{unit-level} surrogacy in the sense of Prentice \cite{prentice1989}, asks whether the outcome depends on treatment only through a surrogate measured on the same observational unit. In the ordinary offline-eval setting considered here, this condition cannot be naturally assessed. The eval scores model outputs while the online metric measures user behavior, and there is no shared unit on which both are observed.

Trial-level surrogacy operates on the objects the decision actually concerns. Both effects are defined at the intervention level, so no shared observational unit is required. The question is simply whether eval effects predict online effects across interventions in the class.  Table~\ref{tab:levels} summarizes the four validity claims.

\begin{table}[tb]
\caption{Four validity claims relevant to an eval suite. Standard eval
validation typically targets the first two. The shipping decision requires
the fourth.}
\label{tab:levels}
\footnotesize
\setlength{\tabcolsep}{4pt}
\begin{tabular*}{\columnwidth}{@{\extracolsep{\fill}}lll@{}}
\toprule
Level & Unit & Typical evidence \\
\midrule
Output-level & model output & judge--human agreement \\
Cross-system & model & benchmark vs.\ preference \\
Unit-level & user within trial & joint $S,Y$ measurement \\
Trial-level & intervention & corpus of paired deltas \\
\bottomrule
\end{tabular*}
\end{table}

\begin{proposition}[Independence of unit-level and trial-level surrogacy]
\label{prop:levels}
(i) There exists a data-generating process in which, within every intervention
$i$, the Prentice conditions hold exactly (in particular
$Y\perp\!\!\!\perp T\mid S$ at the unit level), yet $\tY$ is independent of
$\tS$ across interventions, so that $\Prob(\tY\le0\mid\tS=s)=\Prob(\tY\le0)$ for
every $s$. (ii) There exists a data-generating process in which $\tYi=\tSi$
for every intervention, so that the eval effect determines the sign of the
online effect, yet the Prentice conditions fail within every intervention.
\end{proposition}

\noindent\emph{Proof.} See Appendix~\ref{app:levels}.

The two notions can therefore disagree completely. Mediation within every intervention does not make eval effects predict online effects across interventions, and eval effects can determine online effects even when Prentice's condition fails.

\subsection{Learning from past interventions}
To see why a corpus of past interventions is needed, suppose an oracle supplied, for the new intervention, the distribution of its online effect given its eval delta. The eval alone would then be sufficient, and one could compute the tail probability and decide. In practice this distribution is unknown and must be estimated. A single past intervention yields one eval delta and one outcome delta, which cannot separate the stable relationship of the class from the idiosyncratic draw of that intervention. Multiple past interventions reveal the scatter. Tight scatter after accounting for measurement error indicates a stable conditional that can be transported to a new intervention. Wide scatter indicates that a single pair would have produced a misleading estimate.

A corpus is useful only if the interventions in it share the same eval-to-outcome relationship and the next intervention follows it too. Write $\theta_i=(\tSi,\tYi)^{\top}$ for the true effects of intervention $i$.

\begin{assumption}[Exchangeable corpus]\label{a:corpus}
Within an intervention class $\mathcal{C}$, $\theta_i$ are i.i.d.\ draws from a
distribution $P_{\mathcal{C}}$. The intervention to be decided is a further
independent draw from $P_{\mathcal{C}}$.
\end{assumption}

Assumption~\ref{a:corpus} says that the past interventions in the corpus and the new intervention are draws from the same population of changes. Whatever relationship between eval and online effects holds in the corpus therefore also holds for the next intervention. It is the strongest assumption in the paper, because it requires both that the historical interventions share a common eval-to-outcome relationship and that the next intervention follows the same one. 

\subsection{Measurement noise enters twice}
\label{sec:noise}
The decision target \eqref{eq:target} involves the true eval and online effects, but neither is observed directly. Eval deltas are noisy because the task set and per-task trials are finite, and online deltas are noisy because each A/B test is finite. For each of the $n$ interventions in the historical corpus we observe estimates $\hat\theta_i=(\hat\tau_{S,i},\hat\tau_{Y,i})^{\top}$ of $\theta_i$, and for the new intervention we observe $X$. Noise therefore enters twice, once in the decision and once in the data used to learn it.

In the decision, the team conditions on $X$ rather than on $\tS$. If $\epsilon$ is independent of $\tY$ given $\tS$,
\begin{equation}
\begin{aligned}
&\Prob(\tY\le0\mid X=x)\\
&\quad=\int\Prob(\tY\le0\mid\tS=s)\,dP(s\mid X=x).
\end{aligned}
\label{eq:smooth}
\end{equation}
Conditioning on $X$ averages over the true eval effects that could have produced the observed delta. When the eval is noisy, that range is wide. In the limit $X$ carries no information about the online effect and the sign risk falls back to its class-wide value $\Prob(\tY\le0)$. More historical interventions do not remove this uncertainty. Only a more precise eval does.

The historical corpus has a different noise problem. Noise in the eval deltas flattens the relationship, while noise in the online deltas widens the cloud around it. To learn the relationship between the true effects, we first remove the sampling noise in the corpus. For a new decision, we then put back the uncertainty in the observed eval delta through \eqref{eq:smooth}. The reported measurement variances make the first step possible. Attenuation of the observed correlation \cite{spearman1904} is the second-moment version of the same problem.

Finite-sample noise is centered and shrinks as the task set grows. An LLM judge can also bias the eval effect. If it scores treatment and control outputs differently, the eval effect itself is shifted \cite{panickssery2024,shankar2024}, and more tasks merely make the biased estimate more precise. Section~\ref{sec:judge} returns to this distinction.

\subsection{Observability after deployment}
\label{sec:audit}
Deploying an eval-based decision gate changes which outcomes are observed. Before deployment, the historical corpus contains both an eval effect and an online effect for each intervention. Afterwards, online effects are measured only for the changes sent to an A/B test, so the decisions whose risk the gate is meant to control stop producing outcome labels. In the running example, if large judged wins for the support agent quietly stop converting into resolutions, the data from the remaining A/B tests can look exactly as before. Without audits, drift in that relationship becomes undetectable.

Predeployment validation is therefore not enough. A direct remedy is to send a random fraction of would-be ships, and if needed kills, to audit experiments. Random selection makes those audited outcomes representative of the decisions the gate actually makes and restores labels in the region where the guarantee matters. The appropriate audit rate and test depend on the working model, and Appendix~\ref{app:loop-results} gives one construction.

\section{A Working Model for the Decision}
\label{sec:model}

Section~\ref{sec:setup} defined the decision target without choosing a statistical model. Estimating it requires two further pieces. A measurement model relates the noisy historical deltas to the true intervention effects, and a distributional model turns the recovered relationship into a conditional sign risk. To see what the general claims of Section~\ref{sec:setup} look like in a concrete case, we study the simplest model of each kind, additive measurement error with known variances and bivariate normal effects, stated as three assumptions in addition to Assumption~\ref{a:corpus}. In this model each claim becomes an explicit result that can be computed and checked. 

\begin{assumption}[Moments and measurement error]\label{a:meas}
$P_{\mathcal{C}}$ has mean $\mu=(\mu_S,\mu_Y)^{\top}$, covariance
$\Om=\begin{psmallmatrix}\Oss&\Osy\\\Osy&\Oyy\end{psmallmatrix}\succ 0$, and
finite fourth moments. The estimates satisfy $\hat\theta_i=\theta_i+e_i$. Conditional on $(\theta_i,V_i)$, the
measurement error satisfies $\E[e_i\mid\theta_i,V_i]=0$ and
$\E[e_ie_i^\top\mid\theta_i,V_i]=V_i$, where
$V_i=\mathrm{diag}(\sigma^2_{S,i},\sigma^2_{Y,i})$. The triples
$(\theta_i,V_i,e_i)$ are independent across interventions.
\end{assumption}

\begin{assumption}[Known measurement variances]\label{a:var}
$V_i$ is observed and treated as known for each $i$.
$\bar V_n=n^{-1}\sum_i V_i\xrightarrow{p}\bar V$ finite and
$\sup_i\E\lVert e_i\rVert^4<\infty$.
\end{assumption}

The measurement part of Assumption~\ref{a:meas} is supported by design when the eval and A/B test use disjoint samples, each effect estimator is centered on its true effect, and interventions do not reuse observations. In the typical case, the eval runs on a held-out task set scored offline while the A/B test measures user outcomes on live traffic, so the two estimates share no observations. Disjoint samples make correlated errors less likely but do not rule them out. Appendix~\ref{app:misspec-details} shows that correlated errors can inflate the estimated covariance and false-ship rate. Assumption~\ref{a:var} treats reported measurement variances as known. In practice they are estimated, so this approximation is most credible when each intervention contains enough observations for their estimation error to be negligible. We write $\bar V_S$ and $\bar V_Y$ for the diagonal elements of the limiting average $\bar V$.

Assumptions~\ref{a:corpus}--\ref{a:var} are enough to recover the mean and covariance of the true intervention effects. They are not enough to compute the shipping risks in \eqref{eq:target}. We therefore use normality as a working model for the conditional distribution.

\begin{assumption}[Normal corpus]\label{a:gauss}
$\theta\sim N(\mu,\Om)$ over $\mathcal{C}$ with $\Osy>0$, and
the eval error of the new intervention satisfies $\epsilon\sim N(0,\sigma_S^2)$ independent of $\theta$.
\end{assumption}

If $\Osy<0$, the decision rule reverses, and large positive eval deltas favor kills rather than ships. We state the results for $\Osy>0$. Normal eval error is the usual large-sample approximation for a treatment-effect estimate. Normality of the true effects across interventions is a modeling choice. It implies that the conditional mean of $\tY$ given $\tS$ is linear and that the residual variance is constant across the class. Neither property is argued for on substantive grounds. Both are simplifications inherited from the distributional choice. Appendix~\ref{app:misspec-details} studies departures from this model.

Assumption~\ref{a:corpus} transports the relationship from past interventions to the next one, Assumptions~\ref{a:meas} and~\ref{a:var} let us recover the moments of the true effects from noisy launch logs, and Assumption~\ref{a:gauss} turns those moments into the conditional sign risks the gate needs.

Under Assumption~\ref{a:gauss}, $(X,\tY)$ is bivariate normal, so
\begin{equation}
\begin{aligned}
\tY\mid X=x &\sim N\!\big(m(x),\,v\big),\\
m(x)&=\mu_Y+\tfrac{\Osy}{A}(x-\mu_S),
\end{aligned}
\label{eq:post}
\end{equation}
where $A=\Oss+\sigma_S^2$ and $v=\Oyy-\Osy^2/A$. Equation~\eqref{eq:post} shows exactly what eval noise does. It flattens the slope $\Osy/A$, pulling the predicted online effect back toward the class mean, and it leaves more residual uncertainty about $\tY$. As $\sigma_S^2$ grows, $X$ becomes less informative and the sign risk returns to its class-wide value. Under this model the gate therefore needs only $(\mu,\Om,\sigma_S^2)$.

In this model the strength of the relationship has a scale-free summary,
\begin{equation}
\RC \;=\; \frac{\Osy^2}{\Oss\,\Oyy}\;=\;\Corr(\tS,\tY)^2 .
\label{eq:rc}
\end{equation}
This is the proportion of between-intervention variance in the online effect explained by the best linear predictor based on the eval effect. Under the normal model, the first two moments characterize the full dependence between the two effects. Outside that model, $\RC$ captures only the linear component. Because it discards direction and scale, the shipping rule itself uses $(\mu,\Om,\sigma_S^2)$, not $\RC$ alone.

\paragraph{Running configuration.}
Numerical illustrations use one configuration of the support-agent example throughout. A typical experiment measures the outcome delta with standard error $0.3$pp, and the launch log contains $n=200$ past interventions. The mean eval effect is $\mu_S=0.5$pp with between-intervention variance $\Oss=4$. The mean outcome effect is $\mu_Y=0.05$pp with between-intervention variance $\Oyy=0.09$. We set $\RC=0.6$ and the eval standard error to $\sigma_S=1$pp.

\section{Estimating Surrogacy from Noisy Launch Logs}
\label{sec:estimation}

Under the working model the conditional sign risk depends on the moments $(\mu,\Om)$ of the true intervention effects, but a launch corpus contains only noisy estimates of those effects. In the support-agent example, neither the historical eval improvements nor the historical resolution-rate improvements are observed without noise. Ignoring the noise in either estimate dilutes the observed relationship.

When both $\hat\tau_{S,i}$ and $\hat\tau_{Y,i}$ carry measurement error, the noise inflates the marginal variances. Under Assumption~\ref{a:meas}, the eval and outcome measurement errors are uncorrelated, so measurement error does not alter the covariance in the probability limit, while inflating both marginal variances. The resulting squared correlation is attenuated toward zero.

\begin{proposition}[Attenuation of the naive estimator]
\label{prop:naive}
Under Assumptions~\ref{a:corpus} to \ref{a:var}, the squared sample correlation of the observed deltas converges in
probability to
\begin{equation}
\kappa_S\,\kappa_Y\,\RC,\quad
\kappa_S=\frac{\Oss}{\Oss+\bar V_{S}},\quad
\kappa_Y=\frac{\Oyy}{\Oyy+\bar V_{Y}},
\label{eq:atten}
\end{equation}
the product of the two reliability ratios.
\end{proposition}

\noindent\emph{Proof.} See Appendix~\ref{app:naive-proof}.

In the running configuration, $\kappa_S=4/(4+1)=0.8$ and $\kappa_Y=0.09/(0.09+0.3^2)=0.5$, so the naive estimator converges to $0.8\times0.5\times0.6=0.24$. A small observed offline--online correlation can therefore reflect noisy treatment-effect estimates rather than a weak relationship between the true effects.

The natural correction subtracts the reported within-intervention noise from the sample covariance, as in random-effects estimation \cite{dersimonian1986,danielshughes1997}, recovering the between-intervention covariance up to a positive-semidefinite projection needed in finite samples. We use it rather than maximum likelihood because it isolates the role of measurement error transparently (Appendix~\ref{app:gate-extensions}).

\begin{theorem}[Moment-corrected estimation]
\label{thm:mom}
Let $W_n$ be the sample covariance with divisor $n-1$ of
$\{\hat\theta_i\}_{i=1}^n$. Define the unprojected correction
$\tilde\Om_n=W_n-\bar V_n$ and the estimator
$\hat\Om=\Pi_{\succeq0}(\tilde\Om_n)$, where $\Pi_{\succeq0}$ projects onto
positive semidefinite matrices \cite{higham1988}. Under
Assumptions~\ref{a:corpus} to \ref{a:var}, the following claims hold.
\begin{enumerate}
\item $\E[\tilde\Om_n]=\Om$.
\item $\hat\Om\xrightarrow{p}\Om$.
\item Let
$\psi_i=\mathrm{vech}\!\left((\hat\theta_i-\mu)(\hat\theta_i-\mu)^\top-\Om-V_i\right)$. If some $\delta>0$ satisfies
$\sup_i\E\lVert \psi_i\rVert^{2+\delta}<\infty$ and
$n^{-1}\sum_i\Var(\psi_i)\to\Sigma$, then, for $\Osy\ne0$,
\[
\sqrt{n}(\hat R^2_{\mathcal{C}}-\RC)
\xrightarrow{d}N(0,g^{\top}\Sigma g),
\]
where $g$ is the delta-method gradient given below.
\end{enumerate}
\end{theorem}

\noindent\emph{Proof.} See Appendix~\ref{app:mom-proof}.

The subtraction itself is exactly unbiased. Finite-sample complications enter only through the projection, which is asymptotically inactive when $\Om\succ0$, and through the nonlinear map to $\RC$, whose delta-method gradient in the coordinate order $(\Oss,\Osy,\Oyy)$ is
\begin{equation*}
g=\Big(-\tfrac{\Osy^2}{\Oss^2\Oyy},\;
\tfrac{2\Osy}{\Oss\Oyy},\;
-\tfrac{\Osy^2}{\Oss\Oyy^2}\Big)^{\!\top}.
\end{equation*}
The cross-product moment conditions enter through the asymptotic covariance $\Sigma$.

At $\Osy=0$ the first-order delta-method limit vanishes, so we use the nonparametric bootstrap over interventions for inference. Percentile intervals are slightly too narrow at a few hundred interventions (92--94\% coverage at nominal 95\% in Section~\ref{sec:sims}).

If the launch log contains only interventions that passed an earlier screen such as $\hat\tau_{S,i}\ge c_0$, the moment correction is biased. Appendix~\ref{sec:trunc} gives the corresponding truncated-likelihood correction.

\section{LLM-as-a-Judge Evals}
\label{sec:judge}

The working model of Section~\ref{sec:model} treats finite-sample variation in an eval delta as centered measurement error. An LLM-as-a-judge eval introduces a separate concern. The judge can systematically score treatment and control outputs differently even when the underlying latent quality difference is fixed. Sampling more tasks reduces finite-sample variance, but it does not remove systematic judge distortion. This matters because the gate acts on a treatment-effect estimate rather than an absolute eval score. The same analysis applies when an LLM judge scores live traffic during an A/B test rather than a held-out task set, since the distortion depends on judge behavior rather than on whether scoring happens offline or online.

Let $J_a(t)$ denote the judge verdict for the output on task $t$ under arm $a$. For the binary-score model used here, define sensitivity $Se_a=\Prob(J=1\mid Z=1,a)$ and specificity $Sp_a=\Prob(J=0\mid Z=0,a)$. The observed judge-level eval effect is $\tilde\tau_{S,i}=\E[J_1(t)-J_0(t)]$. We use the following assumption only for the judge-error analysis.

\begin{assumption}[Judging]\label{a:judge}
$(Se_a,Sp_a)$ are constant across tasks within an intervention, and $Se_0+Sp_0>1$.
\end{assumption}

Assumption~\ref{a:judge} says that the judge's sensitivity and specificity are constant across tasks within an intervention and that the control-arm judge is better than random. Define $\lambda_0=Se_0+Sp_0-1$, $\eta_e=Se_1-Se_0$, $\eta_p=Sp_1-Sp_0$, and $\bar\eta=\max(|\eta_e|,|\eta_p|)$. The observed per-arm pass rate is $\tilde p_a=Se_a p_a+(1-Sp_a)(1-p_a)$, where $p_a$ is the latent pass rate.

\paragraph{Arm-homogeneous judge error.}
The benign case is a judge that is equally accurate on both arms. If $Se_1=Se_0=Se$ and $Sp_1=Sp_0=Sp$, then
\[
\tilde\tau_S=\lambda\tS,\qquad \lambda=Se+Sp-1.
\]
If $\lambda$ is common to all interventions, every eval treatment effect is rescaled by the same constant. The sign is preserved when $\lambda>0$, and the trial-level surrogacy computed on the observed judge effects satisfies $\tilde R^2_{\mathcal{C}}=\RC$ whenever $\lambda\ne0$ (Appendix~\ref{app:judge-results}). Judge quality that varies across interventions independently of their effects instead attenuates measured surrogacy (Corollary~\ref{cor:hetero}).

\paragraph{Arm-dependent judge error.}
When sensitivity or specificity differs between treatment and control, the judge effect is no longer a common rescaling. Under Assumption~\ref{a:judge},
\[
\tilde\tau_S=\lambda_0\tS+
\underbrace{\eta_e p_1-\eta_p(1-p_1)}_{b}.
\]
The additive term $b$ is systematic distortion rather than centered sampling noise. For $\tS>0$, the observed judge effect reverses sign exactly when $b<-\lambda_0\tS$, and the distortion obeys $|\tilde\tau_S-\lambda_0\tS|\le\bar\eta$. 

The numerical consequence can be large. In the running judge configuration, sensitivity is $0.90$ and specificity is $0.95$, so $\lambda_0=0.85$. With latent treated pass rate $p_1=0.65$, a judge that is three percentage points less sensitive to treated outputs turns a genuine $+2.0$pp latent eval improvement into an observed $-0.25$pp judge effect. More tasks make the biased conclusion more precise rather than less biased. Self-preference is one documented mechanism for arm-dependent judging \cite{panickssery2024}.

The gate therefore needs a stable eval-effect scale. Dividing by $\lambda$ restores the latent-effect scale under homogeneous error \cite{rogangladen1978}. Arm-dependent bias is not corrected by the moment subtraction, but human-labeled samples from both arms can estimate or bound it. Appendix~\ref{app:judge-results} gives the details.

\section{From Surrogacy to a Shipping Rule}
\label{sec:gate}

We can now return to the original shipping problem. The summary $\RC$ measures the linear strength of trial-level surrogacy, but a gate also needs the direction and scale of the eval--outcome relationship, the class means, and the precision of the new eval estimate. Under the working model these quantities turn \eqref{eq:target} into the normal tail probability $\Prob(\tY\le0\mid X=x)=\Phi[-m(x)/\sqrt v]$. Because $m(x)$ increases when $\Osy>0$, the shipping risk falls as the eval delta grows. The boundary at which that risk reaches $\alpha$ has a closed form.

\begin{theorem}[Oracle gate validity]
\label{thm:gate}
Fix $\alpha\in(0,1/2)$ and define the \emph{surrogate threshold effect}
\begin{equation}
s^{*}(\alpha)=\mu_S+\big(z_{1-\alpha}\sqrt{v}-\mu_Y\big)\,\frac{A}{\Osy},
\label{eq:ste}
\end{equation}
with $z_{1-\alpha}=\Phi^{-1}(1-\alpha)$. Suppose Assumption~\ref{a:gauss} holds and $(\mu,\Om,\sigma_S^2)$ are known.
\begin{enumerate}
\item The policy ships when $X\ge s^{*}(\alpha)$. For every point in this region,
\[
\Prob\big(\tY\le0\mid X=x\big)\le\alpha.
\]
The conditional risk equals $\alpha$ at $x=s^{*}(\alpha)$ and is smaller above the boundary. It follows that
\[
\mathrm{FSR}=\Prob\big(\tY\le0\mid \text{ship}\big)\le\alpha.
\]
\item Define the symmetric lower threshold
\[
s_{\mathrm{lo}}(\alpha)=\mu_S+
\big(-z_{1-\alpha}\sqrt{v}-\mu_Y\big)\frac{A}{\Osy}.
\]
The policy kills when $X\le s_{\mathrm{lo}}(\alpha)$ and satisfies
\[
\Prob(\tY>0\mid\text{kill})\le\alpha.
\]
\end{enumerate}
For $\alpha<1/2$, $s_{\mathrm{lo}}(\alpha)<s^{*}(\alpha)$, and the ship, kill, and A/B test regions partition the line.
\end{theorem}

\noindent\emph{Proof.} See Appendix~\ref{app:gate}.

Theorem~\ref{thm:gate} is an oracle result. It controls the conditional sign risk pointwise when the class parameters are known and the normal model is correct.  The parameter $\alpha$ is not a frequentist significance level at the point null $\tY=0$. It controls the population-averaged sign error under the assumed effect distribution. Following \citet{burzykowski2006}, we use the term surrogate threshold effect for the boundary in \eqref{eq:ste}.

In the running configuration, at $\alpha=0.1$ the oracle gate ships above $s^{*}=2.94$pp, kills below $-3.02$pp, and sends everything between to an A/B test. For the support agent, the shipping threshold is the eval improvement large enough that the probability of lowering true resolution is at most 10 percent. This resolves 14\% of interventions as ships and 6\% as kills without waiting for an online experiment. When the class mean does not already clear the risk target, the threshold rises as the eval--outcome relationship weakens and eventually diverges. Appendix~\ref{app:gate-extensions} gives the exact comparative statics.

In practice the class parameters are estimated rather than known. The plug-in gate substitutes $\hat\mu$ and the corrected $\hat\Om$ from Theorem~\ref{thm:mom}, and sends the intervention to an A/B test if the fitted quantities do not define a positive-slope, positive-variance gate. Its guarantee is weaker than the oracle's. Appendix~\ref{app:gate-extensions} shows that the plug-in gate controls the aggregate false-ship and false-kill rates asymptotically. An asymptotic vanishing-margin construction recovers literal per-decision control 

\paragraph{Two scope conditions.}
The gate is only as valid as the treatment-effect relationship it transports to the next intervention. Two failures are especially important because they can make a seemingly strong gate misleading even when the estimation algebra is correct.

First, the intervention class is part of the estimand. Pooling heterogeneous classes can manufacture trial-level surrogacy. Appendix~\ref{app:pool-results} constructs two classes with $\RC=0$ whose pooled squared correlation is $0.81$, and shows that a pooled gate applied to the weaker class has a realized false-ship rate of $0.106$ at nominal $\alpha=0.05$. Class definitions must therefore be fixed before estimating surrogacy.

Second, deploying the gate makes the eval-to-outcome relationship unobservable in the ship and kill regions (Section~\ref{sec:audit}). Appendix~\ref{app:loop-results} shows that under the working model the class parameters remain identified from the A/B-tested interventions, while the false-ship behavior in the ship region cannot be tested without audits.

\section{Evaluation}
\label{sec:evaluation}

Under the working model of Section~\ref{sec:model}, the theory leaves two practical questions. At corpus sizes a product team might realistically have, does the moment correction recover the underlying relationship? And when the gate uses estimated rather than known parameters, how closely does its realized error track the target rate? We study both in simulation and then run the full pipeline on the Upworthy Research Archive. Stress tests for judge error, pooling, the closed loop, and misspecification are in Appendix~\ref{app:additional}.

\subsection{Simulation study}
\label{sec:sims}
The simulations draw intervention effects from the working model of Section~\ref{sec:model}, so they measure finite-sample behavior. Appendix~\ref{app:misspec-details} studies departures from the model. Unless stated otherwise, effects are in percentage points with
$\mu=(0.5,0.05)$, $\Oss=4$, $\Oyy=0.09$, and measurement variances drawn uniformly within $\pm50$ percent of their mean, with noise-to-signal ratio $\gamma=\bar V_S/\Oss=\bar V_Y/\Oyy$. Appendix~\ref{app:simdetails} gives complete simulation specifications and seeds. The simulation code is available at \url{https://github.com/MSchultzberg/ship-based-on-evals-alone-paper}.

\subsubsection{Recovering surrogacy under measurement noise}
We first ask whether the moment correction removes the attenuation of Proposition~\ref{prop:naive} at a realistic corpus size. Table~\ref{tab:est} reports bias and RMSE for $\hat R^2_{\mathcal{C}}$ at $n=200$ over 1000 replications. The naive squared correlation exhibits the attenuation predicted by \eqref{eq:atten}. The moment-corrected estimator removes most of it, although a noticeable upward finite-sample bias remains in the noisiest cells. That bias comes from the nonlinear ratio and the projection rather than from the covariance correction, which is unbiased by Theorem~\ref{thm:mom}(i), and it shrinks as the corpus grows. At $\RC=0.5$ and $\gamma=1$, nominal 95\% percentile-bootstrap intervals cover 92--94\% of the time at corpus sizes from 100 to 400, with Monte Carlo error about $0.008$.

\begin{table}[t]
\caption{Bias (RMSE) of $\hat R^2_{\mathcal{C}}$ at $n{=}200$, 1000
replications.}
\label{tab:est}
\begin{tabular}{@{}llrr@{}}
\toprule
$\RC$ & $\gamma$ & naive & corrected \\
\midrule
0.2 & 0.5 & $-0.107$ (0.113) & $+0.011$ (0.085) \\
0.2 & 1.0 & $-0.146$ (0.149) & $+0.024$ (0.124) \\
0.2 & 2.0 & $-0.172$ (0.173) & $+0.083$ (0.244) \\
0.5 & 0.5 & $-0.276$ (0.280) & $+0.007$ (0.108) \\
0.5 & 1.0 & $-0.369$ (0.372) & $+0.034$ (0.174) \\
0.5 & 2.0 & $-0.439$ (0.440) & $+0.075$ (0.284) \\
0.8 & 0.5 & $-0.443$ (0.446) & $+0.009$ (0.103) \\
0.8 & 1.0 & $-0.598$ (0.601) & $+0.007$ (0.162) \\
0.8 & 2.0 & $-0.707$ (0.708) & $-0.012$ (0.223) \\
\bottomrule
\end{tabular}
\end{table}

\subsubsection{Gate calibration and finite-sample estimation}
We next ask whether the gate keeps its false-ship and false-kill rates below $\alpha$, first with known parameters and then with parameters estimated from a corpus. For the two-sided oracle gate at $\sigma_S=1$pp, the exact false-ship and false-kill rates lie below their target $\alpha$ in every configuration, and simulations with two million interventions per configuration confirm each value. The threshold falls as surrogacy strengthens, from $6.28$pp at $\RC=0.3$ to $2.35$pp at $\RC=0.9$. At $\RC=0.6$ and $\alpha=0.1$, the threshold is $2.36$, $2.94$, and $5.18$pp when eval noise is $0.5$, $1$, and $2$pp.

The plug-in gate adds parameter-estimation error. Across 1000 corpus draws at $\RC=0.6$ and nominal $\alpha=0.05$, the realized false-ship rate averages $0.033$, $0.030$, and $0.027$ at corpus sizes $n=100$, $200$, and $400$, well below $\alpha$. Individual corpora can still miss. The fraction of corpora for which $\mathrm{FSR}_n>\alpha$ falls from 20\% at $n=100$ to 13\% at $n=200$ and 6\% at $n=400$, and the worst single draw is $0.168$. The costs are asymmetric because underestimating the shipping threshold sends some interventions directly to production, while overestimating it sends them to an A/B test. At corpus sizes a single team is likely to have, parameter estimation is a substantial source of error even though the plug-in gate is valid asymptotically (Appendix~\ref{app:gate-extensions}).

Full oracle operating characteristics, including ship and kill rates, false-ship and false-kill rates, and $\Prob(\mathrm{ship}\mid\tY=0)$, are reported in Appendix~\ref{app:gate-sim}.

\subsection{Upworthy as a Fail-Closed Case}
\label{sec:empirical}
As a public illustration of the intended analysis, we exercise the pipeline on the exploratory split of the Upworthy Research Archive \cite{matias2021}, comprising 4{,}873 of the archive's 32{,}487 experiments. Each test compared two or more headlines for the same article, with randomized traffic allocation and click-through outcomes. The highest-impression arm is the baseline, and we retain one intervention per test, the highest-impression arm among the rest, to satisfy the independence requirement of Assumption~\ref{a:meas}. After removing the 1{,}050 tests created during the documented randomization failure and enforcing an impression floor, the analysis uses 183 independent interventions from 183 tests across 10 headline-feature classes. A class is named by the first text feature the treated headline adds or drops relative to its baseline.

The eval uses an LLM judge (GPT-4o-mini, $K=5$ graded 0--100 ratings per headline) to rate how clearly each headline communicates what the reader will get, without withholding the point to force a click. The eval therefore targets informativeness while the online outcome is click-through engagement. This mismatch makes the application useful for separating measurement precision from trial-level predictive validity. Additional implementation details are reported in Appendix~\ref{app:upworthy-details}.

The estimated reliability ratios are $\hat\kappa_S=0.93$ and $\hat\kappa_Y=0.74$. Their product is $0.68$, so the naive squared correlation is attenuated to about 68 percent of the latent one, and the corrected estimate should be close to the truth if the model holds. Nevertheless, pooled corrected $\hat R^2_{\mathcal{C}}=0.023$. Because Appendix~\ref{app:pool-results} shows that pooling is not a class-specific certificate, we treat this pooled number only as a summary of the public illustration. Per-class samples are small (9--27 interventions), so per-class estimates are exploratory. They range up to $0.15$, and in every class the bootstrap interval for $\Omega_{SY}$ covers zero. For the pooled pipeline illustration, the fitted shipping threshold at $\alpha=0.10$ is $\hat s^{*}=134$pp, beyond the largest possible eval delta of $100$pp, so no intervention ships and every one is routed to an A/B test. Appendix~\ref{app:upworthy-details} reports diagnostics showing that the moment correction is well behaved in this illustration.

This illustrates the intended fail-closed behavior. In this pooled illustration, neither the corrected relationship nor the per-class estimates provide evidence that the eval's treatment effects predict the online treatment effects. A richer rubric, a domain-matched judge, an eval that directly targets engagement, or a more flexible model of the conditional relationship may produce a different result, which would then have to be estimated and validated for that class.

\section{Discussion and Conclusion}
\label{sec:conclusion}

This paper separates two questions that are often conflated. One asks whether an eval is valid. The other asks whether changes in that eval predict changes in the online outcome well enough to support shipping decisions. The second question concerns interventions and has to be answered empirically, from past interventions that are comparable to the next one. Measurement noise can obscure that relationship, judge bias can distort the eval effect itself, and deployment can make the relationship unobservable where it matters most. The Upworthy illustration shows what happens when a precisely measured eval has a weak pooled relationship with the online outcome. The per-class corpora are too small to establish or rule out a class-specific relationship, and the pooled gate ships nothing.

The closed-form guarantee depends on the working model and must be maintained after deployment. Heterogeneous classes, selected launch logs, correlated measurement errors, and latent mixtures each break it in ways analyzed in Appendix~\ref{app:additional}, and the deployed guarantee needs randomized audits to remain testable in the ship region. The linear model also assumes that the exchange rate between eval and outcome effects is constant across a class, which a quality ceiling or a perceptibility threshold would violate. The decision target \eqref{eq:target} and the smoothing in \eqref{eq:smooth} still apply in such cases, but the class moments no longer determine them.

The working model of Section~\ref{sec:model} is the simplest case of the general requirements stated in Section~\ref{sec:setup}. The most important open problem is to develop models with weaker distributional assumptions that can estimate the conditional sign risk, handle measurement-error correction, and detect drift, ideally with asymptotic guarantees that do not depend on bivariate normality. Other open problems include behavior under distribution shift and the choice of audit rate. Another natural extension is the online judge setting, where an LLM scores live sessions during an A/B test. The judge-error analysis of Section~\ref{sec:judge} applies directly, but the shared sessions introduce correlated measurement errors between the judge score and the outcome, combining two threats the present analysis treats separately. An eval can stand in for an online experiment only when its treatment effects have predicted online effects for comparable changes, and when that relationship remains observable after deployment.

\bibliography{refs}

\onecolumn
\appendix
\numberwithin{theorem}{section}
\numberwithin{equation}{section}
\numberwithin{figure}{section}
\numberwithin{table}{section}

\section{Additional Formal Results and Robustness}
\label{app:additional}

This appendix collects the secondary formal results, failure analyses, and stress tests that delimit the scope of the main trial-level-surrogacy argument.

\subsection{Gate properties and plug-in validity}
\label{app:gate-extensions}

This subsection collects further properties of the gate under the working model, including the mean effect of shipped interventions, the comparative statics of the threshold, and the validity of the plug-in rule. Under Assumption~\ref{a:gauss}, the mean realized effect among shipped interventions is
\begin{equation}
\E[\tY\mid\text{ship}]
=\mu_Y+\frac{\Osy}{\sqrt{A}}\cdot
\frac{\phi(\tilde z)}{1-\Phi(\tilde z)},
\label{eq:mills}
\end{equation}
with $\tilde z=(s^{*}-\mu_S)/\sqrt{A}$. In the running configuration, shipped interventions average $+0.38$pp against a class mean of $0.05$pp. This quantity is used as a calibration target in Appendix~\ref{app:loop-results}. The conditional distribution in \eqref{eq:post} is the bivariate-normal analogue of the prediction interval from random-effects meta-analysis \cite{riley2011}, applied to a new intervention rather than a new study.

\begin{proposition}[Monotonicity of the surrogate threshold effect]
\label{prop:ste}
For $\alpha<1/2$ and $\mu_Y<z_{1-\alpha}\sqrt{v}$, $s^{*}(\alpha)$ is
strictly decreasing in $\Osy$ and $s^{*}(\alpha)\to\infty$ as
$\Osy\downarrow0$. Decreasing $\sigma_S^2$ lowers $A$ and $v$, and so lowers $s^{*}(\alpha)$.
\end{proposition}

\noindent\emph{Proof.} See Appendix~\ref{app:gate}.

The threshold therefore falls when the eval--outcome relationship strengthens or the eval becomes more precise, and diverges as the relationship vanishes. The stated condition $\mu_Y<z_{1-\alpha}\sqrt v$ is slightly stronger than the practically relevant one, $\mu_Y<z_{1-\alpha}\sqrt{\Oyy}$, under which the average intervention does not already clear the shipping threshold on prior grounds alone. In the running example, halving eval noise moves the threshold from $2.94$ to $2.36$pp, while raising $\RC$ to $0.9$ moves it to $1.85$pp.

\emph{Remark (operating characteristic at the null).} Under the joint normal model, the conditional distribution of $X$ given $\tY=\delta$ is
\begin{equation*}
X\mid\tY=\delta\sim N\!\left(\mu_S+\frac{\Osy}{\Oyy}(\delta-\mu_Y),\;
A-\frac{\Osy^2}{\Oyy}\right),
\end{equation*}
so the gate's ship probability at any fixed $\delta$ is available in closed form. At $\delta=0$,
\begin{equation*}
\Prob(\mathrm{ship}\mid\tY=0)
=1-\Phi\!\left(
\frac{s^{*}-\mu_S+(\Osy/\Oyy)\mu_Y}
{\sqrt{A-\Osy^2/\Oyy}}
\right).
\end{equation*}
In the running example at $\alpha=0.10$, this probability is $0.047$. The gap between this quantity and $\alpha$ reflects the contribution of the assumed effect distribution to the gate's calibration.

Theorem~\ref{thm:gate} gives the oracle rule when the class parameters are known. In practice they are estimated from the historical corpus. The plug-in gate replaces $(\mu,\Om)$ by the sample mean $\hat\mu=n^{-1}\sum_i\hat\theta_i=(\hat\mu_S,\hat\mu_Y)^\top$ and the noise-corrected covariance $\hat\Om$ from Theorem~\ref{thm:mom}. That substitution introduces estimation error, and Proposition~\ref{prop:plugin} describes what is recovered as the corpus grows.

On the event $D_n=\{\hat\Osy>0,\;\hat A>0,\;\hat v>0\}$, with $\hat A=\hat\Oss+\sigma_S^2$ and $\hat v=\hat\Oyy-\hat\Osy^2/\hat A$, define the plug-in thresholds
\begin{equation}
\hat s^{*}(\alpha)=\hat\mu_S+\big(z_{1-\alpha}\sqrt{\hat v}-\hat\mu_Y\big)\,\frac{\hat A}{\hat\Osy},
\label{eq:plugin}
\end{equation}
and $\hat s_{\mathrm{lo}}(\alpha)$ symmetrically. On $D_n^c$, send the intervention to an A/B test.

The false-ship guarantee requires the normal model to be approximately correct because no amount of data compensates for a misspecified conditional distribution. It also requires a corpus large enough that estimation error does not inflate the rate beyond $\alpha$. Appendix~\ref{app:misspec-details} tests the first requirement. Proposition~\ref{prop:plugin} addresses the second.

Define the aggregate false-ship rate $\mathrm{FSR}_n=\Prob(\tY\le0\mid X\ge\hat s^{*}(\alpha),\hat\mu,\hat\Om)$ on $D_n$ and $\mathrm{FSR}_n=0$ on $D_n^c$, and the false-kill rate $\mathrm{FKR}_n$ symmetrically.

\begin{proposition}[Asymptotic validity of the plug-in gate]
\label{prop:plugin}
Suppose Assumptions~\ref{a:corpus}--\ref{a:var} and \ref{a:gauss} hold, with $\Om\succ0$, $\Osy>0$, and fixed $\alpha\in(0,1/2)$, and let the new intervention be an independent draw from $P_{\mathcal{C}}$, independent of the historical corpus. Then the following statements hold.

\noindent (i) $\hat s^{*}(\alpha)\xrightarrow{p}s^{*}(\alpha)$ and $\hat s_{\mathrm{lo}}(\alpha)\xrightarrow{p}s_{\mathrm{lo}}(\alpha)$.

\noindent (ii) $\Prob\{\mathrm{FSR}_n>\alpha\}\to0$ and $\Prob\{\mathrm{FKR}_n>\alpha\}\to0$.

\noindent (iii) $\sup_{x\ge\hat s^{*}(\alpha)}\Prob(\tY\le0\mid X=x)\xrightarrow{p}\alpha$.
\end{proposition}

\noindent\emph{Proof.} See Appendix~\ref{app:plugin}.

Proposition~\ref{prop:plugin} says that the plug-in thresholds converge to the oracle thresholds (i), the probability that the aggregate false-ship or false-kill rate exceeds $\alpha$ vanishes (ii), and the worst-case pointwise risk converges to $\alpha$ (iii). Part (iii) does not imply that the plug-in rule satisfies the pointwise level-$\alpha$ inequality with probability tending to one. The oracle rule in Theorem~\ref{thm:gate} achieves pointwise control because the parameters are exact. The plug-in rule achieves only aggregate control, where estimation error is absorbed by the gap between the boundary risk and the average risk over all shipped interventions. This gap exists because conditional risk falls beyond the boundary, so its average over shipped interventions is below $\alpha$. Appendix~\ref{app:plugin} makes this precise. In the simulation, the probability that a single corpus produces $\mathrm{FSR}_n > \alpha$ falls from 20 percent at $n=100$ to 6 percent at $n=400$. A team that needs a per-corpus statement at these sizes can use the vanishing-margin construction below or a smaller nominal $\alpha$. Iteration can erode this slack when teams stop as soon as the eval clears the threshold, concentrating shipped interventions near the boundary.

\begin{corollary}[Vanishing-margin feasible gate with pointwise control]
\label{cor:margin}
Under the conditions of Proposition~\ref{prop:plugin} and the moment conditions of Theorem~\ref{thm:mom}(iii), let $\alpha_n=\alpha-\delta_n$ with $\delta_n>0$, $\delta_n\to0$, and $\sqrt{n}\,\delta_n\to\infty$. Construct the plug-in gate at level $\alpha_n$ but evaluate it against $\alpha$. Then
\[
\Prob\big\{\Prob(\tY\le0\mid X=x)\le\alpha\;\text{for every}\;x\ge\hat s^{*}(\alpha_n)\big\}\to1,
\]
and similarly for the kill gate. Because $\alpha_n\to\alpha$, the thresholds converge to the same oracle limits, with $\hat s^{*}(\alpha_n)\xrightarrow{p}s^{*}(\alpha)$.
\end{corollary}

\noindent\emph{Proof.} See Appendix~\ref{app:plugin}.

The corollary recovers the pointwise guarantee by targeting a slightly stricter level whose margin dominates the $n^{-1/2}$ estimation error. The conservatism vanishes, so the limiting gate is identical. A team that wants a literal per-decision statement uses $\alpha_n$. Aggregate control needs only $\alpha$.

\paragraph{Alternative estimation under the normal model.}
Under Assumption~\ref{a:gauss}, with conditionally normal historical measurement errors and noninformative $V_i$, maximum likelihood \cite{vanhouwelingen2002,riley2007} can jointly estimate $(\mu,\Om)$ with heterogeneous $V_i$, handle the legacy-gate truncation of Appendix~\ref{sec:trunc} without a separate moment correction, and put estimation and decision under the same model. The main analysis uses the moment estimator because it is easier to reason about and isolates the measurement-error correction.

\subsubsection{Gate operating characteristics}
\label{app:gate-sim}
Table~\ref{tab:gate} shows how the oracle gate trades resolved decisions against errors as surrogacy and $\alpha$ vary. The entries are exact, and each is confirmed by two million simulated interventions.

\begin{table*}[t]
\centering
\caption{Two-sided gate at $\sigma_S{=}1$pp with known parameters. Every entry is exact and confirmed by $2{\times}10^6$ simulated interventions per row. FSR is the false-ship rate, FKR the false-kill rate, both
bounded by $\alpha$ by Theorem~\ref{thm:gate}. $\Prob(\mathrm{ship}\mid\tY=0)$ is the closed-form operating characteristic at the null.}
\label{tab:gate}
\begin{tabular}{@{}llrrrrrr@{}}
\toprule
$\RC$ & $\alpha$ & $s^{*}$(pp) & ship & kill & FSR & $\Prob(\mathrm{ship}\mid\tY=0)$ & FKR \\
\midrule
0.3 & 0.05 & 6.28 & 0.005 & 0.001 & 0.036 & 0.001 & 0.038 \\
0.3 & 0.10 & 4.84 & 0.026 & 0.004 & 0.071 & 0.010 & 0.075 \\
0.3 & 0.20 & 3.09 & 0.124 & 0.033 & 0.138 & 0.078 & 0.149 \\
0.6 & 0.05 & 3.79 & 0.071 & 0.025 & 0.025 & 0.014 & 0.027 \\
0.6 & 0.10 & 2.94 & 0.137 & 0.058 & 0.049 & 0.047 & 0.054 \\
0.6 & 0.20 & 1.92 & 0.263 & 0.132 & 0.097 & 0.149 & 0.109 \\
0.9 & 0.05 & 2.35 & 0.203 & 0.111 & 0.015 & 0.033 & 0.017 \\
0.9 & 0.10 & 1.85 & 0.273 & 0.160 & 0.030 & 0.080 & 0.034 \\
0.9 & 0.20 & 1.23 & 0.371 & 0.235 & 0.061 & 0.187 & 0.071 \\
\bottomrule
\end{tabular}
\end{table*}

\FloatBarrier
\subsection{Additional LLM-as-a-Judge Results}
\label{app:judge-results}

Judge error matters because the gate reads eval effects at face value. Arm-homogeneous error rescales eval effects, while arm-dependent error can shift and even reverse them. We use the notation of Section~\ref{sec:judge} throughout.

\subsubsection{Balanced error is benign}
A judge that is equally accurate on both arms changes the scale of eval effects but nothing else.

\begin{proposition}[Balanced judge error]
\label{prop:nondiff}
If $Se_1=Se_0=Se$ and $Sp_1=Sp_0=Sp$, then
$\tilde\tau_S=\lambda\,\tS$ with $\lambda=Se+Sp-1$. Consequently, if
$\lambda$ is common to all interventions, every trial effect is rescaled by
the same constant, the sign of every effect is preserved when $\lambda>0$, and, when
$\lambda\ne0$, the surrogacy computed on observed effects $\tilde R^2_{\mathcal{C}}$ equals the latent $\RC$.
\end{proposition}

\noindent\emph{Proof.} See Appendix~\ref{app:judge-proof}.

In the running example, $\lambda=0.85$ simply shrinks every eval delta by 15 percent while leaving $\RC$ exactly at $0.6$. Dividing by $\lambda$ restores the latent-effect scale \cite{rogangladen1978}, although it cannot recover statistical power lost to judge noise. What matters is whether that scaling is stable across the corpus. If judge quality changes across interventions independently of their effects, for example because the judge version, prompt, or criteria drift, the measured surrogacy is attenuated and must be re-anchored with human labels.

\begin{corollary}[Judge heterogeneity across interventions]
\label{cor:hetero}
Suppose $\tilde\tau_{S,i}=\lambda_i\tSi$ with $\lambda_i$ independent of
$\theta_i$, $\E[\lambda_i]=\bar\lambda\ne 0$, $\Var(\lambda_i)=
\sigma_\lambda^2$. Then trial-level surrogacy on observed effects is
\begin{equation}
\tilde R^2_{\mathcal{C}}
=\frac{\RC}{\,1+\dfrac{\sigma_\lambda^2}{\bar\lambda^2}\cdot
\dfrac{\E[\tS^2]}{\Oss}\,}\;\le\;\RC .
\label{eq:hetero}
\end{equation}
\end{corollary}

\noindent\emph{Proof.} See Appendix~\ref{app:judge-proof}.

When this variation is unrelated to the effects it attenuates measured surrogacy irreversibly, and when it is related it can distort it in either direction.

\subsubsection{Arm-dependent error can reverse sign}

When judge accuracy differs between treatment arms, the situation is qualitatively worse. In the running example with a latent treated pass rate $p_1=0.65$, a judge three percentage points less sensitive to outputs in the treated system's style turns a genuine $+2.0$pp improvement into an observed $-0.25$pp. The judge can therefore report a negative eval effect for an intervention that genuinely improves latent quality. The bias can also correlate with intervention type, model swaps inducing it while prompt tweaks do not, which contaminates $\Osy$ itself. Self-preference \cite{panickssery2024} is a documented mechanism for this arm dependence.

Recall that $\bar\eta=\max(|\eta_e|,|\eta_p|)$ is the magnitude of the arm dependence.

\begin{theorem}[Arm-dependent error, sign reversal, partial identification]
\label{thm:diff}
Under Assumption~\ref{a:judge},
\begin{equation}
\tilde\tau_S=\lambda_0\,\tS+\underbrace{\eta_e\,p_1-
\eta_p\,(1-p_1)}_{b}.
\label{eq:diff}
\end{equation}
(i) For $\tS>0$, $\tilde\tau_S<0$ if and only if $b<-\lambda_0\tS$.
(ii) $|\tilde\tau_S-\lambda_0\tS|\le\bar\eta$.
(iii) The identified set for $\tS$ given $(\lambda_0,\bar\eta)$ is
\begin{equation}
\tS\in\Big[\frac{\tilde\tau_S-\bar\eta}{\lambda_0},\;
\frac{\tilde\tau_S+\bar\eta}{\lambda_0}\Big]\cap[-1,1],
\label{eq:pid}
\end{equation}
and every point of the set is attainable provided
$\bar\eta\le\min(Se_0,1-Se_0,Sp_0,1-Sp_0)$.
\end{theorem}

\noindent\emph{Proof.} See Appendix~\ref{app:diff-proof}.

Theorem~\ref{thm:diff} characterizes how arm-dependent judge error distorts the eval delta. Part (i) gives the exact condition under which the observed delta reverses sign. Part (ii) bounds the total bias by $\bar\eta$, with a sharper bound when only sensitivity or specificity differs across arms. Part (iii) gives the identified set for the latent eval effect when only a bound on arm dependence is available. It shows how wide the uncertainty gets when per-arm accuracy is bounded rather than estimated.

In the running example, the identified set \eqref{eq:pid} for the $+2.0$pp intervention is $[-3.8,+3.2]$pp. The observed eval delta carries almost no information about the sign of the true effect. With 1000 tasks per arm and latent control pass rate $p_0=0.60$, a sensitivity deficit of $3$pp flips the sign of the observed eval delta 58 percent of the time at $\tS=1.6$pp. More tasks make the biased conclusion more precise, not less biased.

\subsubsection{Diagnosing arm dependence in practice}

Human labels on outputs from both arms identify each arm's sensitivity and specificity, and hence the arm dependence. When per-arm $(Se_a,Sp_a)$ are estimated from human-labeled samples, \eqref{eq:pid} becomes a point correction and the bias $b$ is removed. When only the bound $\bar\eta$ is credible, no point correction is available and the identification interval governs what the eval delta can say.

If a team has human-labeled (input, output) pairs from both arms for an LLM-judged component, scored on the same rubric the judge uses, it can estimate $Se_a$ and $Sp_a$ per arm and compute $\hat\eta_e = \hat{Se}_1 - \hat{Se}_0$ and $\hat\eta_p = \hat{Sp}_1 - \hat{Sp}_0$. Standard confidence intervals for these differences of proportions give an upper bound on $\bar\eta$. Within the single-score model used here, substituting that bound into the identification interval of Theorem~\ref{thm:diff}(iii) shows whether the observed eval delta remains informative about the true effect.

\subsection{Pooling Across Intervention Classes}
\label{app:pool-results}

The intervention class $\mathcal{C}$ is part of the scientific specification
of the estimand, not a variance-reduction choice. All probability statements
in the preceding sections are conditional on a fixed class definition
specified before estimating the parameters used for the gate. Automatic class
discovery from the data is outside scope.

The corpus an organization actually has mixes prompt tweaks, retrieval changes, ranking
adjustments, and model swaps. Pooling across them changes the estimand, and Proposition~\ref{prop:pool} shows why.

\begin{proposition}[Pooling manufactures surrogacy]
\label{prop:pool}
Assume $\beta\ne0$.
Let two classes occur with equal probability, with within-class covariances
$\Om^{(c)}=\mathrm{diag}(\omega_S,\omega_Y)$ (so $\RC=0$ in each class) and
class means $\mu^{(1)}=-(d,\beta d)$, $\mu^{(2)}=(d,\beta d)$. By the law of
total covariance the pooled covariance is
$\Om_{\mathrm{pool}}=\mathrm{diag}(\omega_S,\omega_Y)+
\begin{psmallmatrix}d^2&\beta d^2\\\beta d^2&\beta^2d^2\end{psmallmatrix}$,
hence
\begin{equation}
R^2_{\mathrm{pool}}
=\frac{\beta^2d^4}{(\omega_S+d^2)(\omega_Y+\beta^2d^2)}\;\to\;1
\quad(d\to\infty).
\end{equation}
With $\omega_S=\omega_Y=\beta=1$ and $d=3$, $R^2_{\mathrm{pool}}=0.81$ while every class has $\RC=0$.
\end{proposition}

\noindent\emph{Proof.} See Appendix~\ref{app:pool-proof}.

Proposition~\ref{prop:pool} says that pooling heterogeneous intervention classes can produce strong apparent surrogacy even when the eval has no trial-level surrogacy within either class. In the pooling simulation of Appendix~\ref{app:pool-sim}, classes with $\RC$ of $0.72$ and $0.02$ pool to $0.27$, a number that describes neither class.

A pooled corpus can therefore certify an eval that predicts nothing within any
class the team actually decides in, a Simpson-type failure. The converse
misuse also occurs. A pooled slope applied to a class whose own slope differs
miscalibrates the gate, and Appendix~\ref{app:pool-sim} shows a doubling of the
false-ship rate from this mistake. Cross-model benchmark validation is
the extreme case, pooling across systems separated by capability gulfs and
then transporting to marginal product changes. Any team using the framework
must define intervention classes before estimating surrogacy and never pool
across them.

We propose two tests to detect pooling problems before they inflate the false-ship rate. First, homogeneous pools should have the same eval-to-outcome slope across candidate classes. This can be tested by fitting the slope per class and using a standard Wald test. Second, the estimated conditional mean should predict the held-out intervention means. For held-out $i$, regress $\hat\tau_{Y,i}$ on $m(\hat\tau_{S,i})$ with weights $1/(\sigma^2_{Y,i}+v_i)$, where $v_i=\Oyy-\Osy^2/(\Oss+\sigma^2_{S,i})$, and test $(\text{intercept},\text{slope})=(0,1)$ by a Wald test. Rejection means the estimated eval-to-outcome relationship does not hold for new interventions from the same class.

\subsubsection{Simulation of Pooled Intervention Classes}
\label{app:pool-sim}
Proposition~\ref{prop:pool} is a stylized construction, so we also simulate two classes that mimic prompt tweaks and model swaps, with class-level surrogacy $0.723$ and $0.022$. The pooled corpus reports $0.271$. Applying the pooled gate at nominal $\alpha=0.05$ to class-2 interventions ships 9.8\% of them at a realized false-ship rate of $0.106$, more than twice the nominal level. The correctly class-conditioned gate sets $s^{*}=13.5$pp for class 2 and sends every change to an experiment. Figure~\ref{fig:pool} shows the two effect clouds and the misleading pooled slope.

\begin{figure}[t]
\centering
\includegraphics[width=0.72\textwidth]{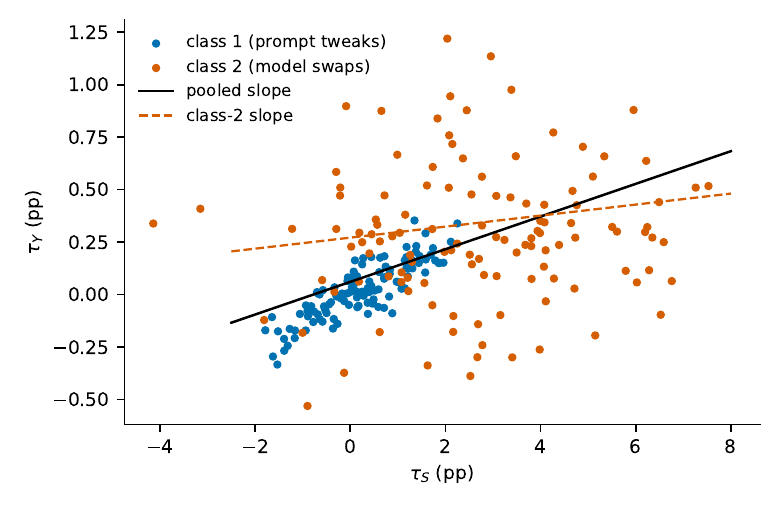}
\caption{Two intervention classes with pooled and class-specific regression slopes. The pooled slope both overstates surrogacy for class 2 and miscalibrates its gate.}
\label{fig:pool}
\end{figure}

\FloatBarrier
\subsection{Closed-loop identification and auditing}
\label{app:loop-results}

Deploying the gate changes the data available over time to evaluate and re-estimate it. New online treatment-effect estimates are observed only for interventions sent to an A/B test. Interventions above the shipping threshold are shipped on eval evidence alone, so no outcome estimate is collected for the decisions whose false-ship risk the gate is intended to control. Without deliberate auditing, the eval-to-outcome relationship can therefore degrade in the ship region without appearing in the data used to monitor the gate. In the running example, experiments run only for observed deltas between $-3.0$ and $+2.9$pp. If large judged wins quietly stop converting into resolutions, the harm falls on interventions shipped above $+2.9$pp, where no experiment runs.

Formally, throughout this section candidates arrive as an i.i.d.\ stream from the class model. Outcome labels are $\hat\tau_Y=\tY+e_Y$, with $e_Y$ mean zero, variance $\sigma_Y^2$, and independent of $(\tS,\tY,\epsilon)$. Labels arrive only on the A/B test region $B=(s_{\mathrm{lo}},s^{*})$.
Let $n_B$ denote the number of labeled interventions observed in $B$.

\begin{proposition}[Identification and blindness under self-selection]
\label{prop:selfsel}
(i) $(\mu,\Om)$ is identified from the marginal distribution of $X$ and the conditional distribution of $\hat\tau_Y\mid X$ on $B$.

\noindent (ii) For any $r\in[0,1]$, there exists a joint distribution $Q$ that agrees with the working model on $B$ but has $\Prob_Q(\tY\le0\mid X\ge s^{*})=r$. No test based on the unaudited data can distinguish $Q$ from the working model.

\noindent (iii) If $e_Y$ is also normal, a labeled point at $X=x$ adds $(x-\bar x)^2/(v+\sigma_Y^2)$ to the Fisher information for the slope $\Osy/A$, up to $O(n_B^{-1})$.
\end{proposition}

\noindent\emph{Proof.} See Appendix~\ref{app:selfsel-proof}.

Proposition~\ref{prop:selfsel} says that all parameters are identified from the data the decision rule generates (i), but the false-ship guarantee in the ship region cannot be tested from the unaudited data (ii). This contrasts with the selective labels problem \cite{lakkaraju2017}, where selection acts on unobservables and identification itself fails. Here selection acts on the observed eval delta, so identification holds but falsifiability does not. Part (iii) says that labeled points far from the A/B test region are the most informative for estimating the slope, which motivates auditing shipped interventions.

A simple way to detect a poorly calibrated gate is to run an A/B test on a random fraction of the interventions the rule would otherwise ship without one. These are the \emph{audit experiments}, and the fraction sent to audit is the \emph{audit rate} $q$.

Since the audit selection is independent of $(\tY, X)$ conditional on the ship region, non-audited ships retain the guarantee of the underlying gate. Theorem~\ref{thm:gate} gives the oracle guarantee, and Proposition~\ref{prop:plugin} gives the plug-in guarantee.

\begin{proposition}[The audited gate]
\label{prop:audit}
Let the audit rate $q$ be constant on the ship region. Then the following statements hold.

\noindent (i) The audited outcome labels are i.i.d.\ with mean $\E[\tY\mid X\ge s^{*}]$.

\noindent (ii) The audited sample identifies the conditional distribution of $\hat\tau_Y$ given $X$ on the ship region.
\end{proposition}

\noindent\emph{Proof.} See Appendix~\ref{app:audit-proof}.

Proposition~\ref{prop:audit} says that the audited outcomes have a known predicted mean (i), which can be compared against \eqref{eq:mills} to test whether the model holds in the ship region. Under normal measurement errors, the standardized residuals $(\hat\tau_{Y,i}-m(X_i))/\sqrt{v+\sigma_{Y,i}^2}$ are i.i.d.\ standard normal, giving an exact test under known parameters. In the running example, audited ships should average $+0.38$pp. Part (ii) extends the calibration test of Appendix~\ref{app:calib} to the ship region.

\emph{How small can the audit rate be?} Every claim in
Proposition~\ref{prop:audit}(i) holds for any $q>0$, so the rate primarily buys speed. The leverage result in Proposition~\ref{prop:selfsel}(iii) makes
audits efficient, though audited experiments with unbalanced splits need more observations than full experiments to reach the same precision. Standardized audit residuals
are i.i.d.\ standard normal under known parameters, so any anytime-valid
changepoint detector converts them into a revert trigger with a lifetime
false-alarm budget. The simulations use a changepoint mixture of one-sided normal-mixture e-processes, nonnegative supermartingales whose reciprocals give anytime-valid $p$-values, updated after every audited outcome (Appendix~\ref{app:simdetails}). 

That false-alarm bound assumes the parameters are fixed in advance. Our simulation refits them every epoch, so the nominal budget is only a rule of thumb. Estimating the parameters on a separate data set does not by itself solve the problem, because a fixed estimation error can still accumulate in a one-sided detector. A fully valid implementation would need to account for parameter uncertainty, for example through uniformly valid confidence sequences. In the simulation the realized pre-break false-trip rate is $0.025$, inside the $0.05$ budget. At the example's volume, an audit rate of $0.4$ detects the simulated reward-hacking break in 96 percent of runs by the horizon, at a median of 3.5 epochs after it begins (Appendix~\ref{app:loop-sim}). The framework benefits most from a high win rate. More ships mean even a low audit rate produces enough labeled outcomes to detect drift, and the gate saves the most experiments. At a low win rate, few ships occur, a high audit rate is needed to accumulate any audited outcomes at all, and the gate is saving little anyway. In both regimes the audited A/B tests run alongside the ship, so the cost is a small traffic split rather than a delay. The optimal rate remains a value-of-information choice we leave open.

The audit is necessary because two cheaper monitors cannot substitute for it. A sequential monitor on the eval stream catches shifts in the eval-delta marginal, and residuals from the A/B test region catch shifts in the eval-to-outcome relationship at moderate deltas. Neither can detect a slope change confined to the ship region. That is the reward-hacking case, where changes improve the eval score without improving the user outcome. The break simulated in Appendix~\ref{app:loop-sim} is not even confined to the ship region, since it acts on the true eval effect and noisy eval deltas carry some affected interventions into the A/B test region. A sequential monitor on the A/B-test residuals nonetheless detects it in only 1.7 percent of runs by the horizon, and after the break no nontrivial rule that ships on eval evidence alone can meet the promised rate.

\subsubsection{Simulation of Closed-Loop Degradation}
\label{app:loop-sim}
The closed-loop simulation asks what an unnoticed reward-hacking break costs over a finite horizon, with and without audits. The break occurs at epoch seven. From then on the conditional mean of $\tY$ given the true eval effect $\tS$ bends at $\tS=2$pp and falls with slope $-0.08$ beyond that point, while the distribution of eval effects is unchanged. Because the gate acts on the noisy delta $X$, some affected interventions still fall in the A/B test region, so the A/B-test data do change, but only slightly. Without audits the pipeline ships about 19 of 150 candidates per epoch at a realized false-ship rate of $0.31$ against the promised $0.10$, roughly 6 harmful ships per epoch. A sequential monitor on the residuals of the A/B-tested interventions detects the break in only 1.7 percent of runs by the horizon, although its own pre-break false-trip rate is 4.7 percent. With audit rate $0.4$, the same detector on audit residuals detects it in 96 percent of runs, at a median of 3.5 epochs after the break, with a pre-break false-trip rate of 2.5 percent. Before the break, every pipeline realizes a false-ship rate near $0.056$.

Pooling audits into a trailing-window refit instead of testing them reduces harmful ships over the broken epochs from 59 to 31 per run, compared with 23 under the audited e-detector, but its false-ship rate rises to $0.36$. No nontrivial rule that ships on eval evidence alone can do better than $0.21$ after the break, because that is the smallest conditional sign risk at any eval delta. Figure~\ref{fig:loop} summarizes the trajectories.

\begin{figure}[t]
\centering
\includegraphics[width=0.76\textwidth]{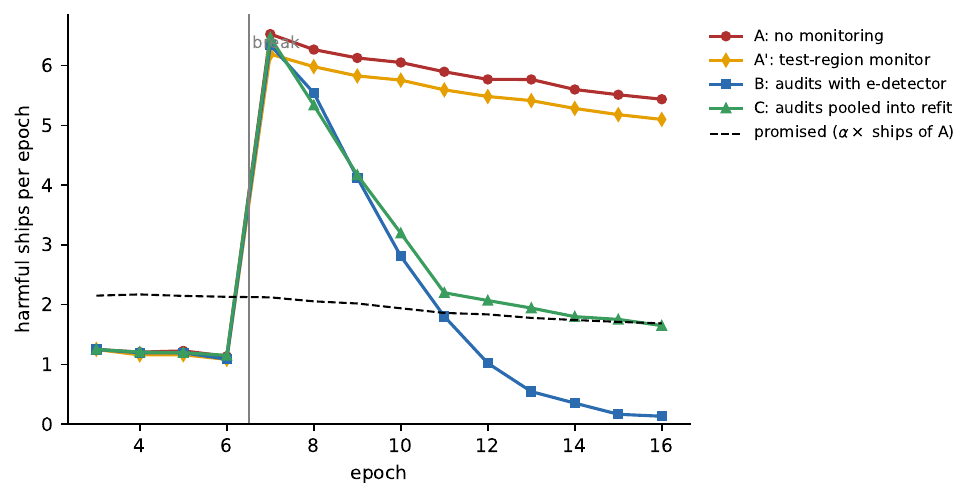}
\caption{Harmful ships per epoch under a reward-hacking break at epoch seven. Without audits the pipeline ships at about three times its promised risk for the rest of the horizon, and a monitor on the A/B-test residuals rarely notices. The audited pipeline's e-detector trips at a median of 3.5 epochs after the break and reverts the class to online testing.}
\label{fig:loop}
\end{figure}

\FloatBarrier
\subsection{Misspecification details}
\label{app:misspec-details}
The simulations in Section~\ref{sec:sims} draw from the model the estimator assumes.
Table~\ref{tab:misspec} perturbs the assumptions one at a time. Each cell defines its own population, so $\RC$ is
measured from that population rather than assumed, and the realized false-ship
rate of each fitted gate is read off that population, with a gate that is not defined counting as a false-ship rate of zero, as in Proposition~\ref{prop:plugin} (Appendix~\ref{app:simdetails}). Assumption~\ref{a:judge} is not perturbed
here, since arm-dependent judging needs latent quality and per-arm human
labels rather than a simulation choice.

Assumption~\ref{a:gauss} splits in two. Its population half, that true effects are normal across the class, is tested via two single-component departures: heavy-tailed $t_5$ effects and a right-skewed eval marginal. Both leave the estimator accurate and the gate conservative. This robustness does not extend to latent mixtures, where the estimator recovers $\RC$ but the gate fails. Departures from normal measurement error have a larger effect, though not enough in this simulation to threaten the guarantee. A $t_5$ shape for the historical measurement errors at the reported variance inflates $\hat R^2_{\mathcal{C}}$ from $0.60$ to $0.63$ while the realized rate stays at $0.060$. Only the new eval error enters Assumption~\ref{a:gauss}, so this row tests the sensitivity of the estimator rather than the gate's distributional assumption. A heavy-tailed error puts most of its draws below its average, so subtracting the average $\bar V_n$ removes too much in the typical sample, and the projection rate rises from $0.25$ to $2.0$ percent. Errors with no fourth moment, such as $t_3$, violate Assumption~\ref{a:var} outright and are outside the theory.

Computing $V_i$ correctly requires accounting for clustered questions and within-question resampling dependence \cite{miller2024}. Understating the measurement variance fails conservatively in this experiment. Reported standard errors that understate the true noise by 20 percent leave both variances inflated, $\RC$ reads $0.34$ against a truth of $0.60$, and the gate answers by raising the threshold and shipping 4 percent of candidates instead of 15. The cost is forgone throughput rather than risk.

Assumption~\ref{a:meas} fails dangerously. The off-diagonal of the sample covariance is unbiased only because the eval and outcome errors are uncorrelated. Disjoint data make this plausible but do not guarantee it, and when it fails the covariance inherits the error correlation. At correlation $0.3$ the estimate reads $0.86$, the threshold falls, and the realized rate reaches $0.096$, still inside the promise but with the margin
gone. At $0.5$ it reads $0.98$ and the realized FSR reaches $0.116$. Any pipeline
that computes the eval and the outcome from overlapping data, for instance by
scoring the same sessions the experiment measures, is in this cell.

A latent mixture is the departure that does hurt. When one intervention in five is drawn from a second component with a weak and offset link, the corpus and the new interventions are still drawn from one common distribution, so Assumption~\ref{a:corpus} holds and only normality fails. The estimator gets the pooled $\RC$ right, $0.341$ against a truth of $0.327$, and the gate still realizes $0.262$. Getting the pooled number right is no help when the conditional distribution varies across the mixture components. The class and held-out calibration tests can expose heterogeneity along observed partitions but do not validate conditional tails, so a latent mixture can evade them. Randomized ship-region audits can reveal realized miscalibration after deployment, but are safeguards rather than ex ante certificates. Assumption~\ref{a:corpus} is not perturbed anywhere in the table.

Two diagnostics flag unstable moment correction. The first is the ratio $\sigma_Y^2/v$, which measures outcome noise relative to residual variance. Values much above one indicate that the subtracted noise dominates the signal. The second is the fraction of bootstrap resamples in which $W_n - \bar V_n$ fails to be positive semidefinite. These diagnostics address noise subtraction, not conditional-model misspecification such as the latent mixture.

\FloatBarrier
\begin{table}[htbp]
\caption{Each row perturbs one assumption. $n{=}200$, $\alpha{=}0.10$, 2000
corpora per row (\texttt{code/simulations.ipynb}, Section 8). Values above $0.10$ violate the nominal FSR bound.}
\label{tab:misspec}
\centering
\begin{tabular}{@{}llrrrr@{}}
\toprule
perturbation & assumption broken & true $\RC$ & mean $\hat R^2_{\mathcal{C}}$ & ship rate & realized FSR \\
\midrule
none & none & 0.600 & 0.614 & 14.7\% & 0.052 \\
heavy-tailed effects ($t_5$) & \ref{a:gauss}, effects & 0.600 & 0.614 & 13.1\% & 0.060 \\
skewed eval effects & \ref{a:gauss}, effects & 0.600 & 0.616 & 13.0\% & 0.042 \\
heavy-tailed errors ($t_5$) & none (shape sensitivity) & 0.601 & 0.630 & 15.6\% & 0.060 \\
s.e.\ understated 20\% & \ref{a:var} & 0.600 & 0.343 & 4.1\% & 0.022 \\
correlated errors, $\rho{=}0.3$ & \ref{a:meas} & 0.600 & 0.864 & 26.0\% & 0.096 \\
correlated errors, $\rho{=}0.5$ & \ref{a:meas} & 0.600 & 0.977 & 31.1\% & 0.116 \\
latent two-component mixture & \ref{a:gauss}, effects & 0.327 & 0.341 & 3.9\% & 0.262 \\
\bottomrule
\end{tabular}
\end{table}

\FloatBarrier

\subsection{Corpora that were already gated}
\label{sec:trunc}
Most launch logs are not a clean sample of what a team tried. If an online
experiment was run only after some informal offline check had passed, the
recorded interventions entered the corpus under a rule such as
$\hat\tau_{S,i}\ge c_0$, and the estimator of Theorem~\ref{thm:mom} is biased.
It is tempting to assume the damage is one-sided, since a gated log holds a
narrower range of eval deltas and narrow ranges usually make a relationship
look weaker. That is not what happens. The estimator subtracts the measurement
noise you report, and on a gated log the remaining spread can be smaller than
the noise being removed, so the subtraction takes away more than the data
contains. The arithmetic then returns something that cannot be a covariance
matrix at all, the correction repairs it, and a repaired matrix looks like
perfect surrogacy. With unit effect variances, $\RC=0.6$, measurement variances $\sigma_S^2=\sigma_Y^2=0.5$, and a legacy gate at the mean, this is exactly what happens, and $\RC$ reads $1$
instead of something below $0.6$. So a gated log can make an eval look stronger
rather than weaker, and the gate built on it can ship too readily.

The repair is a likelihood correction. For this correction, condition on the
observed $V_i$ and assume
$\theta_i\mid V_i\sim N(\mu,\Om)$ and
$e_i\mid(\theta_i,V_i)\sim N(0,V_i)$. The first condition states that $V_i$
is noninformative about the true effects. The density of
$(\hat\tau_{S},\hat\tau_{Y})$ conditional on $V_i$ among included
interventions is then the truncated bivariate normal
\begin{equation*}
\frac{f_{N}\big((\hat\tau_S,\hat\tau_Y);\mu,\Om+V_i\big)}
{\Prob\big(\hat\tau_S\ge c_0;\mu_S,\Oss+\sigma^2_{S,i}\big)}\,
\mathbf{1}\{\hat\tau_S\ge c_0\},
\end{equation*}
and maximizing over $(\mu,\Om)$ with the $V_i$ held at their known values
restores consistency. A legacy gate that let interventions through with some probability is handled the same way, by weighting with the known rule. When the legacy gate selects on a vector of grader scores rather than a single eval delta, the multivariate selection formulas of \citet{muthen1989,muthen2015} apply. This correction may be necessary for many real corpora, and it needs the old rule to be known, or at least recoverable from the log, which is the practical limit on using it.

\subsection{Additional Upworthy implementation details}
\label{app:upworthy-details}
Outcome data are split binomially within each test so that the gate's decisions can be checked against independent outcome estimates. Counting negative held-out estimates among ships measures $\Prob(\hat\tau_Y\le0\mid\text{ship})$ rather than the false-ship rate, so the sign errors of the held-out estimates must be accounted for before the two are compared. Because the fitted gate ships nothing in this illustration, the held-out half goes unused. The positive-semidefinite projection is never needed, neither for the pooled estimate nor in any of 2{,}000 bootstrap resamples, no class estimate lies on the boundary, and outcome noise is $0.37$ times the residual variance, so the moment correction is well behaved. The notebook \texttt{code/upworthy.ipynb} reproduces every estimate in this illustration from the public archive and a cache of the judge's ratings.

\section{Proofs}
\label{app:proofs}

\subsection{Proof of Proposition~\ref{prop:levels}}
\label{app:levels}
\emph{(i)} Within intervention $i$, let $T\in\{0,1\}$ denote the treatment assignment, and at the
unit level $S=\delta_iT+\zeta$, $Y=b_iS+\nu$, with $\zeta,\nu$ independent
mean-zero noises independent of $T$. Across interventions, let $\delta_i$ and
$c_i$ be independent, with $\delta_i$ uniform on $\{1,2\}$ and $c_i$ uniform on
$\{-1,1\}$, and set $b_i=c_i/\delta_i$. Because $Y$ is a function of $(S,\nu)$
with $\nu\perp(T,S)$, $Y\perp T\mid S$, which is Prentice's conditional
independence criterion. In every intervention treatment affects $S$
($\delta_i\ne0$), $S$ is prognostic for $Y$ ($b_i\ne0$), and treatment affects
$Y$ ($\tYi=b_i\delta_i=c_i\ne0$). The trial effects are $\tSi=\delta_i$ and
$\tYi=c_i$, which are independent by construction. Hence
$\Prob(\tY\le0\mid\tS=s)=\Prob(c_i\le0)=1/2$ for every $s$ in the support of $\tS$.
\emph{(ii)} Let $S=\delta_iT+\zeta$ and $Y=\delta_iT+\nu$ with independent
nondegenerate normal noises, so the conditional supports of $S$ under both arms overlap, and $\delta_i$ supported away from zero. Then
$\tSi=\tYi=\delta_i$, so the eval effect determines the online effect. Within every intervention,
$\E[Y\mid T,S]=\delta_iT$ depends on $T$
after conditioning on $S$ (which is a noisy measurement of $\delta_iT$), so
$Y\not\perp T\mid S$ and Prentice's criterion fails. \hfill$\square$

\subsection{Proof of Proposition~\ref{prop:naive}}
\label{app:naive-proof}
Assumption~\ref{a:meas} gives centered measurement errors with zero
cross-coordinate covariance. Assumption~\ref{a:var} gives
$\bar V_n\xrightarrow{p}\bar V$ and the moment bound needed for the law of
large numbers. Therefore the sample moments converge to
\begin{gather*}
\Cov(\hat\tau_S,\hat\tau_Y)\to\Osy,\qquad
\Var(\hat\tau_S)\to\Oss+\bar V_S,\qquad
\Var(\hat\tau_Y)\to\Oyy+\bar V_Y,
\end{gather*}
and the continuous mapping theorem gives \eqref{eq:atten}. \hfill$\square$

\subsection{Proof of Theorem~\ref{thm:mom}}
\label{app:mom-proof}
\emph{Part (i).} Let $X_i=\hat\theta_i$. Assumption~\ref{a:meas} gives
$\E[X_i]=\mu$ and
$\Var(X_i)=\Om+\E[V_i]$. For the divisor-$(n-1)$ sample covariance,
\begin{align*}
\E[W_n]
&=\frac{1}{n-1}\Big(\sum_i\Var(\hat\theta_i)-n\Var(\bar{\hat\theta})\Big)\\
&=\Om+\frac{1}{n}\sum_i\E[V_i].
\end{align*}
Assumption~\ref{a:var} makes every $V_i$ available for the subtraction.
Since $\E[\bar V_n]=n^{-1}\sum_i\E[V_i]$, it follows that
$\E[\tilde\Om_n]=\Om$.

\emph{Part (ii).} The fourth-moment bound in Assumption~\ref{a:var}, together
with the finite fourth moments in Assumption~\ref{a:meas}, gives the law of
large numbers for the sample covariance. The same assumptions give
$\bar V_n-n^{-1}\sum_i\E[V_i]\xrightarrow{p}0$. Therefore
$\tilde\Om_n\xrightarrow{p}\Om$. Since $\Om\succ0$ and eigenvalues are
continuous, the truncation $\Pi_{\succeq0}$ is the identity on a neighborhood
of $\Om$, so $\hat\Om\xrightarrow{p}\Om$ and the projection is inactive with
probability tending to one.

\emph{Part (iii).} With $\psi_i$ defined in the theorem,
\[
\sqrt{n}\,\mathrm{vech}(\tilde\Om_n-\Om)
=n^{-1/2}\sum_i \psi_i+o_p(1).
\]
The remainder comes from centering the sample covariance at its sample mean.
The stated $(2{+}\delta)$-moment bound gives the Lyapunov condition, and
$n^{-1}\sum_i\Var(\psi_i)\to\Sigma$ fixes the limiting covariance. The
Lindeberg--Feller central limit theorem therefore applies. Part (ii) shows
that the projection contributes nothing asymptotically. The
map $\Om\mapsto\Osy^2/(\Oss\Oyy)$ is differentiable at $\Om\succ0$ with the
stated gradient
\[
\frac{\partial}{\partial\Osy}=\frac{2\Osy}{\Oss\Oyy},\qquad
\frac{\partial}{\partial\Oss}=-\frac{\Osy^2}{\Oss^2\Oyy},\qquad
\frac{\partial}{\partial\Oyy}=-\frac{\Osy^2}{\Oss\Oyy^2},
\]
and the delta method concludes. \hfill$\square$

\subsection{Proof of Theorem~\ref{thm:gate},
Proposition~\ref{prop:ste}}
\label{app:gate}
Under Assumption~\ref{a:gauss}, $(X,\tY)$ is jointly normal with
$\Var(X)=A$, $\Cov(X,\tY)=\Osy$, giving \eqref{eq:post}. The conditional mean
$m(x)$ is strictly increasing in $x$ (as $\Osy>0$) with constant variance
$v$, so $x\mapsto\Prob(\tY\le0\mid X=x)=\Phi(-m(x)/\sqrt v)$ is strictly
decreasing. Solving $m(s^{*})=z_{1-\alpha}\sqrt v$ gives \eqref{eq:ste}, at
which the conditional non-positivity probability equals $\alpha$ exactly, and
it is below $\alpha$ for all $x>s^{*}$. The aggregate bound follows from the tower property.
\[
\Prob(\tY\le0\mid X\ge s^{*})
=\E\big[\Prob(\tY\le0\mid X)\mid X\ge s^{*}\big]\le\alpha.
\] The lower gate is
the mirror argument. For \eqref{eq:mills}, the normal identity (Stein's lemma applied to the indicator, or direct integration of the bivariate normal)
\[
\E[\tY\mathbf{1}\{X>s\}]=\mu_Y\Prob(X>s)+
\frac{\Cov(\tY,X)}{\sqrt{\Var X}}\,\phi\!\Big(\frac{s-\mu_S}{\sqrt{\Var X}}\Big)
\]
divided by $\Prob(X>s)=1-\Phi(\tilde z)$ yields \eqref{eq:mills}. For the proposition, write $c=\Osy$, so
$s^{*}-\mu_S=(z\sqrt v-\mu_Y)A/c$ with $v=\Oyy-c^2/A$, and differentiate,
\begin{equation*}
\frac{\partial s^{*}}{\partial c}
=-\frac{z}{\sqrt v}-\frac{A\,(z\sqrt v-\mu_Y)}{c^{2}}.
\end{equation*}
Both terms are negative when $z>0$, that is $\alpha<1/2$, and
$\mu_Y<z\sqrt v$, and as $c\downarrow0$, $v\to\Oyy$ and
$s^{*}\to\infty$. The same computation gives
$\partial s^{*}/\partial\mu_S=1$ and
$\partial s^{*}/\partial\Oss=(z\sqrt v-\mu_Y)/\Osy
+z\,\Osy/(2A\sqrt v)>0$ in the same regime. \hfill$\square$

\subsection{Proof of Proposition~\ref{prop:plugin} and Corollary~\ref{cor:margin}}
\label{app:plugin}

\emph{Part (i).} Under Assumptions~\ref{a:corpus} to \ref{a:var},
$\hat\theta_i=\theta_i+e_i$ with $\E[e_i\mid\theta_i]=0$ and uniformly
bounded fourth moments, so
$\Var(\bar{\hat\theta})=n^{-2}\sum_i\Var(\hat\theta_i)=O(n^{-1})$ and
Chebyshev gives $\hat\mu\xrightarrow{p}\mu$. Theorem~\ref{thm:mom}(ii)
gives $\hat\Om\xrightarrow{p}\Om$. Since $\Om\succ0$ and $\Osy>0$, the
quantities $A=\Oss+\sigma_S^2>0$ and $v=\Oyy-\Osy^2/A>0$ (using
$\Osy^2/A\le\Osy^2/\Oss<\Oyy$ by positive definiteness) are strictly
positive, and $\Osy>0$. The maps $\Om\mapsto\Osy$,
$\Om\mapsto\Oss+\sigma_S^2$, and
$\Om\mapsto\Oyy-\Osy^2/(\Oss+\sigma_S^2)$ are continuous at $\Om$, so
$\Prob(D_n)\to1$ by the continuous mapping theorem. On a neighborhood of the
true $(\mu,\Om)$ all denominators in \eqref{eq:plugin} are bounded away from
zero and the square-root argument is strictly positive, so $s^{*}$ is a
continuous function of $(\mu,\Om)$ at the true parameters. The continuous
mapping theorem applied to $(\hat\mu,\hat\Om)\xrightarrow{p}(\mu,\Om)$
gives $\hat s^{*}(\alpha)\xrightarrow{p}s^{*}(\alpha)$, and symmetrically
for $\hat s_{\mathrm{lo}}$.

\emph{Part (ii).} Define $g(x)=\Prob(\tY\le0\mid X=x)$. Under
Assumption~\ref{a:gauss}, $\tY\mid X=x\sim N(m(x),v)$ with
$m'(x)=\Osy/A>0$, so
\begin{equation}
g(x)=\Phi\big(-m(x)/\sqrt v\big)
\label{eq:gx}
\end{equation}
is continuous and strictly decreasing. By construction,
$m(s^{*}(\alpha))=z_{1-\alpha}\sqrt v$, so
$g(s^{*}(\alpha))=\Phi(-z_{1-\alpha})=\alpha$.

Define the aggregate false-ship function
\begin{align*}
H(s)&=\Prob(\tY\le0\mid X\ge s)=\E\big[g(X)\mid X\ge s\big]\\
&=\frac{\int_s^{\infty}g(x)\,f_X(x)\,dx}{1-F_X(s)},
\end{align*}
where $X\sim N(\mu_S,A)$, so $f_X>0$ everywhere and the denominator is
strictly positive for every finite $s$. Both numerator and denominator are
continuous in $s$, hence $H$ is continuous.

Since $g$ is strictly decreasing and $g(s^{*})=\alpha$, we have
$g(x)<\alpha$ for every $x>s^{*}$. The continuous distribution of $X$
puts zero mass on the boundary, so
\begin{equation}
H(s^{*})=\E[g(X)\mid X\ge s^{*}]<\alpha.
\label{eq:hslack}
\end{equation}

The new intervention is independent of the historical corpus. Conditional
on $(\hat\mu,\hat\Om)$, the fitted threshold is a fixed number and the
new $(X,\tY)$ retains its population distribution. On $D_n$,
$\mathrm{FSR}_n=H(\hat s^{*}(\alpha))$. By part~(i), continuity of $H$,
and the continuous mapping theorem,
$H(\hat s^{*}(\alpha))\xrightarrow{p}H(s^{*}(\alpha))$. Together with
$\Prob(D_n)\to1$,
\[
\mathrm{FSR}_n\xrightarrow{p}H(s^{*}(\alpha))<\alpha.
\]

Let $\Delta=\alpha-H(s^{*}(\alpha))>0$ by \eqref{eq:hslack}. Then
$\Prob(\mathrm{FSR}_n>\alpha)\le\Prob(|\mathrm{FSR}_n-H(s^{*})|>\Delta)
\to0$.

For the kill gate, define $k(x)=\Prob(\tY>0\mid X=x)=\Phi(m(x)/\sqrt v)$,
which is continuous and strictly increasing. By construction
$k(s_{\mathrm{lo}}(\alpha))=\alpha$, and
$L(s)=\E[k(X)\mid X\le s]$ is continuous with
$L(s_{\mathrm{lo}})<\alpha$ by the same boundary-measure argument.
Consistency of $\hat s_{\mathrm{lo}}$ and continuity of $L$ give
$\mathrm{FKR}_n\xrightarrow{p}L(s_{\mathrm{lo}})<\alpha$ and
$\Prob(\mathrm{FKR}_n>\alpha)\to0$.

\emph{Part (iii).} Since $g$ is strictly decreasing,
$\sup_{x\ge\hat s^{*}}g(x)=g(\hat s^{*})$. By part~(i) and continuity of $g$,
$g(\hat s^{*})\xrightarrow{p}g(s^{*})=\alpha$. However, convergence
\emph{to} $\alpha$ does not imply $\Prob\{g(\hat s^{*})\le\alpha\}\to1$.
An estimator can converge to a boundary value while remaining above it with
nonvanishing probability. The aggregate result in part~(ii) works because
$H(s^{*})<\alpha$ strictly, providing a gap that absorbs estimation error.
No such gap exists at the pointwise boundary. \hfill$\square$

\emph{Proof of Corollary~\ref{cor:margin}.} Under the additional moment conditions, Theorem~\ref{thm:mom}(iii) gives
$\sqrt{n}\,\mathrm{vech}(\hat\Om-\Om)=O_p(1)$, and the mean satisfies
$\sqrt{n}\,(\hat\mu-\mu)=O_p(1)$ by the variance calculation in part~(i).
The threshold $s^{*}$ is continuously differentiable in $(\mu,\Om,\alpha)$
at the true parameters (all denominators nonzero, square-root argument
positive), so the mean-value theorem gives
$\hat s^{*}(\alpha_n)-s^{*}(\alpha_n)=O_p(n^{-1/2})$.

Since $g$ is continuously differentiable and strictly decreasing with
$g(s^{*}(\alpha))=\alpha$, and $g'(s^{*}(\alpha))=-\phi(z_{1-\alpha})\,
\Osy/(A\sqrt v)<0$, the mean-value theorem applied to
$g(s^{*}(\alpha_n))=\alpha_n=\alpha-\delta_n$ gives
\begin{align*}
s^{*}(\alpha_n)-s^{*}(\alpha)
&=\frac{\delta_n}{|g'(s^{*}(\alpha))|}+o(\delta_n)\\
&=\frac{A\sqrt v}{\phi(z_{1-\alpha})\,\Osy}\,\delta_n+o(\delta_n).
\end{align*}
The deterministic gap is of order $\delta_n$, and the estimation error is
$O_p(n^{-1/2})$. Since $\sqrt{n}\,\delta_n\to\infty$, we have
$n^{-1/2}=o(\delta_n)$, so
$\hat s^{*}(\alpha_n)-s^{*}(\alpha)=
[s^{*}(\alpha_n)-s^{*}(\alpha)]+O_p(n^{-1/2})$,
which is positive with probability tending to one. On that event, every
$x\ge\hat s^{*}(\alpha_n)$ satisfies $x\ge s^{*}(\alpha)$, and since
$g(x)\le g(s^{*}(\alpha))=\alpha$ for all such $x$, the pointwise
guarantee holds. Because $\alpha_n\to\alpha$, continuity gives
$s^{*}(\alpha_n)\to s^{*}(\alpha)$ and hence
$\hat s^{*}(\alpha_n)\xrightarrow{p}s^{*}(\alpha)$.

The kill gate follows symmetrically. The function $k$ is strictly increasing with
$k(s_{\mathrm{lo}}(\alpha))=\alpha$ and $k'(s_{\mathrm{lo}})>0$, so
targeting $\alpha_n<\alpha$ yields
$s_{\mathrm{lo}}(\alpha_n)<s_{\mathrm{lo}}(\alpha)$ by a margin of order
$\delta_n$, dominating the $O_p(n^{-1/2})$ estimation error, and on the
high-probability event $\hat s_{\mathrm{lo}}(\alpha_n)\le s_{\mathrm{lo}}(\alpha)$, every $x\le\hat s_{\mathrm{lo}}(\alpha_n)$ satisfies
$k(x)\le k(s_{\mathrm{lo}}(\alpha))=\alpha$. \hfill$\square$

\subsection{Proof of Proposition~\ref{prop:nondiff} and
Corollary~\ref{cor:hetero}}
\label{app:judge-proof}
\[
\tilde p_1-\tilde p_0
=Se(p_1-p_0)+(1-Sp)\big((1-p_1)-(1-p_0)\big)=(Se+Sp-1)(p_1-p_0),
\]
giving the first claim, and a common positive scaling of one coordinate leaves both the
correlation and its square unchanged. For the corollary, with
$\tilde\tau_{S,i}=\lambda_i\tSi$ with $\lambda_i\perp\theta_i$,
\begin{align*}
\Cov(\tilde\tau_S,\tY)&=\bar\lambda\,\Osy, \\
\Var(\tilde\tau_S)&=\E[\lambda^2]\E[\tS^2]-\bar\lambda^2(\E[\tS])^2
=\bar\lambda^2\Oss+\sigma_\lambda^2\,\E[\tS^2],
\end{align*}
and substituting into \eqref{eq:rc} gives \eqref{eq:hetero}. \hfill$\square$

\subsection{Proof of Theorem~\ref{thm:diff}}
\label{app:diff-proof}
Expanding $\tilde p_1-\tilde p_0$ with $Se_1=Se_0+\eta_e$,
$Sp_1=Sp_0+\eta_p$,
\begin{align*}
\tilde\tau_S
&=Se_0(p_1-p_0)+\eta_ep_1
+(1-Sp_0)\big((1-p_1)-(1-p_0)\big)-\eta_p(1-p_1)\\
&=\lambda_0\tS+\eta_ep_1-\eta_p(1-p_1).
\end{align*}
(i) is immediate from \eqref{eq:diff}. For (ii), by the triangle inequality
\[
|b|\le|\eta_e|\,p_1+|\eta_p|\,(1-p_1)
\le\bar\eta\,p_1+\bar\eta\,(1-p_1)=\bar\eta.
\]
If
$\eta_p=0$ then $|b|=|\eta_e|p_1\le\bar\eta p_1
\le\bar\eta\max(p_1,1-p_1)$, symmetrically if $\eta_e=0$. For (iii), sharpness follows because for any
$b_0\in[-\bar\eta,\bar\eta]$, take
$\eta_e=b_0$, $\eta_p=-b_0$ (within the bound, and valid
probabilities under the stated condition), for which
$b=b_0p_1+b_0(1-p_1)=b_0$ for every $p_1$, so the bias attains every value in
the interval, and inverting \eqref{eq:diff} for $\tS$ yields exactly
\eqref{eq:pid}. For feasibility of the intersection, for any candidate $\tau\in[-1,1]$ take
$p_0=\max(0,-\tau)$ and $p_1=p_0+\tau$, both in $[0,1]$, and no
$\tau\notin[-1,1]$ is a difference of probabilities. \hfill$\square$

\subsection{Proof of Proposition~\ref{prop:pool}}
\label{app:pool-proof}
For pooling, with class label $c$ uniform on $\{1,2\}$, the law of total
covariance gives
$\Cov(\tS,\tY)=\E[\Cov(\tS,\tY\mid c)]+\Cov(\E[\tS\mid c],\E[\tY\mid c])
=0+\beta d^2$ and analogously
$\Var(\tS)=\omega_S+d^2$, $\Var(\tY)=\omega_Y+\beta^2d^2$, from which the
displayed formula and limit follow, and the numeric instance is direct
substitution. \hfill$\square$

\subsection{Proof of Proposition~\ref{prop:selfsel}}
\label{app:selfsel-proof}
(i) $\mu_S=\E[X]$ and $A=\Var(X)$ involve only the unlabeled stream.
Since $e_Y$ is independent of $(\tS,\tY,\epsilon)$ and mean zero,
$\E[\hat\tau_Y\mid X]=\E[\tY\mid X]=\mu_Y+(\Osy/A)(X-\mu_S)$ and
$\Var(\hat\tau_Y\mid X)=v+\sigma_Y^2$ under Assumption~\ref{a:gauss}.
Both conditional moments are unchanged by conditioning on the
$X$-measurable event $\{X\in B\}$, and a nondegenerate interval $B$
identifies slope, intercept, and residual variance, hence
$\Osy=\text{slope}\cdot A$, $\mu_Y$, $v$, $\Oyy=v+\Osy^2/A$, and
$\Oss=A-\sigma_S^2$. Consistency follows from the law of large numbers and
continuity of the maps. (ii) Construct $Q$ by keeping the distribution of $X$ and the
conditional distribution of $\tY$ given $X=x$ for $x<s^{*}$, and replacing it for
$x\ge s^{*}$ by a two-point distribution placing mass $r$ on $-1$ and $1-r$ on $+1$. The observed data, $\{X_i\}$ together with
$\{\hat\tau_{Y,i}:X_i\in B\}$, is a measurable function of quantities
whose joint distribution is unchanged, so its pushforward is identical under the two
distributions and every test's rejection probability coincides. (iii) With unknown intercept the per-observation information matrix is
\[
\frac{1}{s^2}\begin{pmatrix}1 & x-\mu_S\\ x-\mu_S & (x-\mu_S)^2\end{pmatrix},
\qquad s^2=v+\sigma_Y^2,
\]
and the efficient information for the slope from a design $\{x_i\}$ is $s^{-2}(S_2-S_1^2/n_B)$ with $S_k=\sum_i(x_i-\mu_S)^k$. Adding a point at $x$ raises it by $(x-\bar x)^2/\{s^2(1+n_B^{-1})\}$. Estimating the slope from one region alone would instead give $\Var(X\mid X\in R)/s^2$, which is small in the tail because tail points cluster, but that is not the right comparison here, because the audits are added to a design that the A/B test region has already pinned down.
\hfill$\square$

\subsection{Proof of Proposition~\ref{prop:audit}}
\label{app:audit-proof}
Let $\xi\in\{0,1\}$ indicate selection for an audit experiment, with
$\Prob(\xi=1\mid X\ge s^{*})=q$. The guarantee for non-audited ships follows from independence. On $\{X\ge s^{*}\}$ with constant $q$, $\xi$ is independent of
$(\tY,X)$, so conditioning on $\xi=0$ changes nothing and
Theorem~\ref{thm:gate} applies. For general $q(\cdot)$,
$\Prob(\tY\le0\mid X\ge s^{*},\xi=0)$ is a $(1-q(X))$-weighted average
over $\{X\ge s^{*}\}$ of $g(x)=\Prob(\tY\le0\mid X=x)\le\alpha$,
hence at most $\alpha$.
(i) With constant $q$ the indicator is independent of $(X,\tY)$ on $\{X\ge s^{*}\}$, so audited ships are i.i.d.\ draws from the shipped distribution and $\E[\hat\tau_Y]=\E[\tY\mid X\ge s^{*}]$, which is \eqref{eq:mills}. Asymptotic normality follows from the central limit theorem and consistency from the law of large numbers. Truncating a bivariate normal on $X$ leaves $\hat\tau_Y$ skew-normal, so the mean is not exactly normal in finite samples. The residual $\tY-m(X)$ is normal under Assumption~\ref{a:gauss} and independent of the $X$-measurable selection, and adding $e_Y\sim N(0,\sigma_Y^2)$ keeps it normal, which is where that condition is used. For $X$-dependent $q$ the audited density carries the factor $q(x)/\E[q(X)\mid X\ge s^{*}]$. (ii) The argument of
Proposition~\ref{prop:selfsel}(i) applies with $B$ replaced by
$[s^{*},\infty)$. \hfill$\square$

\section{Linear-Projection Calibration Test}
\label{app:calib}
This test checks whether the linear projection of the outcome on the eval transports to held-out interventions. Its null hypothesis holds for any joint distribution with the assumed first two moments, so passing it is necessary evidence for transport of the fitted regression relationship but does not validate the conditional mean function, normal conditional tails, or the gate guarantee under misspecification.

For held-out interventions $i=1,\dots,m$, condition on the observed $V_i$ and,
in addition to Assumption~\ref{a:meas}, assume
$\theta_i\mid V_i\sim N(\mu,\Om)$ and
$e_i\mid(\theta_i,V_i)\sim N(0,V_i)$. Without these conditional normal models,
the distribution below is an approximation rather than an identity. Define
$A_i=\Oss+\sigma^2_{S,i}$,
$m_i=\mu_Y+\frac{\Osy}{A_i}(\hat\tau_{S,i}-\mu_S)$, and
$v_i=\Oyy-\Osy^2/A_i$ (the heterogeneous-noise analogue of the gate-time residual variance $v$), so that
$\hat\tau_{Y,i}\mid\hat\tau_{S,i},V_i\sim N(m_i,\;v_i+\sigma^2_{Y,i})$
with the per-intervention eval noise entering both the predictor and its
variance. Weighted least squares of $\hat\tau_{Y,i}$ on $(1,m_i)$ with
weights $w_i=(v_i+\sigma^2_{Y,i})^{-1}$, design matrix $M$, and weight matrix
$W=\mathrm{diag}(w_i)$ yields coefficients
$\hat\beta_{\mathrm{cal}}=(M^{\top}WM)^{-1}M^{\top}W\hat\tau_Y$ that, conditional on
$M$ and the observed $V_i$, are exactly normal with mean $(0,1)^{\top}$ and
covariance $(M^{\top}WM)^{-1}$ under the null of calibrated predictions. With
known first-stage parameters, the Wald quadratic form
\[
(\hat\beta_{\mathrm{cal}}-(0,1)^{\top})^{\top}(M^{\top}WM)(\hat\beta_{\mathrm{cal}}-(0,1)^{\top})
\]
follows $\chi^2_2$. With plug-in
$(\hat\mu,\hat\Om)$ from a disjoint training split, first-stage estimation
error inflates the variance, and the trial-level bootstrap over the training
split delivers the reference distribution. \hfill$\square$

\section{Simulation Details}
\label{app:simdetails}
The notebooks \texttt{code/simulations.ipynb} and \texttt{code/upworthy.ipynb} reproduce every number in the paper from fixed seeds. Quantities with a closed form or a one-dimensional integral are computed exactly. Monte Carlo is used only where the paper reports a simulation experiment. The simulation notebook draws all randomness from one master seed (20261004) split into independent streams, and every parameter is documented in the cell that uses it. The gated-corpus example of Appendix~\ref{sec:trunc} is computed exactly from the moments of the truncated bivariate normal.

\end{document}